%% file: main.tex
\documentclass{article}

\input{commands}

\begin{document}

{
  \title{\bf Inference of multivariate exponential Hawkes processes with inhibition and application to neuronal activity}
  \author{Anna Bonnet\thanks{Université de Paris and Sorbonne Université, CNRS, Laboratoire de Probabilités, Statistique et Modélisation, F-75013 Paris, France}, Miguel Martinez Herrera\footnotemark[1], Maxime Sangnier\footnotemark[1]}
  \maketitle
}

\begin{abstract}
The multivariate Hawkes process is a past-dependent point process used to model the relationship of event occurrences between different phenomena.
Although the Hawkes process was originally introduced to describe excitation effects, which means that one event increases the chances of another occurring, there has been a growing interest in modelling the opposite effect, known as inhibition.
In this paper, we focus on how to infer the parameters of a multidimensional exponential Hawkes process with both excitation and inhibition effects. Our first result is to prove the identifiability of this model under a few sufficient assumptions. Then we propose a maximum likelihood approach to estimate the interaction functions, which is, to the best of our knowledge, the first exact inference procedure in the frequentist framework.
Our method includes a variable selection step in order to recover the support of interactions and therefore to infer the connectivity graph.
A benefit of our method is to provide an explicit computation of the log-likelihood, which enables in addition to perform a goodness-of-fit test for assessing the quality of estimations.
We compare our method to standard approaches, which were developed in the linear framework and are not specifically designed for handling inhibiting effects.
We show that the proposed estimator performs better on synthetic data than alternative approaches. We also illustrate the application of our procedure to a neuronal activity dataset, which highlights the presence of both exciting and inhibiting effects between neurons.


\end{abstract}

\section{Introduction}

A Hawkes process is a point process in which each point is commonly associated with event occurrences in time. In this past-dependent model, every event time impacts the probability that other events take place subsequently. These processes are characterised by the conditional intensity function, seen as an instantaneous measure of the probability of event occurrences. Since their introduction in \cite{Hawkes1971}, Hawkes processes have been applied in a wide variety of fields, for instance in seismology \citep{Ogata1988}, social media \citep{Rizoiu2017}, criminology \citep{Short2020} and neuroscience \citep{Reynaud2018}.

The multidimensional version of this model, referred to as the multivariate Hawkes process, describes the appearance of different types of events, the occurrences of which are influenced by all past events of all types. Each interaction between two types of events is encoded in kernel functions, also called interaction functions.
Originally this model takes only into account mutually exciting interactions - an event increases the chances of others occurring -
by assuming that all kernel functions are non-negative.
A specificity of self-exciting Hawkes processes is their branching structure, also known as cluster structure.
Introduced in \cite{Hawkes1974}, this parallel between Hawkes processes and branching theory has provided the first theoretical background for the self-exciting Hawkes model, in particular existence and expected number of points on a finite interval.
Estimation methods in the literature are vast including maximum likelihood estimators \citep{Ozaki1979,Guo2018} and method of moments \citep{DaFonseca2013}. 
Nonparametric approaches include an EM procedure introduced in \cite{Mohler2011}, estimations obtained via the solution of Wiener-Hopf equations \citep{Bacry2016} or by approximating the process through autoregressive models \citep{Kirchner2017} or through functions in reproducing kernel Hilbert spaces \citep{Yang2017}.

Although the self-exciting Hawkes process remains widely studied, there has been a growing interest in modeling the opposite effect, known as inhibition, in which the probability of observing an event is lowered by the occurrence of certain events. In practice, this amounts to considering negative kernel functions. In order to maintain the positivity of the intensity function, a non-linear operator is added to the expression which in turns entails the loss of the cluster representation. This model known as the non-linear Hawkes process was first presented in \cite{Bremaud1996}, where existence of such processes was proved via construction using bi-dimensional marked Poisson processes. Such approach of analysis has been used in the literature as in \cite{Chen2017}, where a coupling process is established and leveraged to obtain theoretical guarantees on cross-analysis covariance. Another approach is presented in \cite{Costa2020}, where renewal theory allows to obtain limit theorems for processes with bounded support kernel functions. Estimation methods focus mainly on nonparametric methods for general interactions and non-linear functions, as found in \cite{Bacry2016,Sulem2021}.

In the last years,
alternative models have been designed in order to take into account inhibiting effects in Hawkes processes.
An example is the neural Hawkes process, presented in \cite{Mei2017,Zuo2021}, which combines a multivariate Hawkes process and a recurrent neural network architecture.
In \cite{Duval2021}, a multiplicative model considers two sets of neuronal populations, one exciting and another inhibiting, and each intensity function is the product of two non-linear functions (one for each group).
Another model is presented in \cite{Short2020} and called self-limiting Hawkes process.
It includes the inhibition as a multiplicative term in front of a the traditional self-exciting intensity function. 

In this paper, we present a maximum likelihood estimation method for multivariate Hawkes processes with exponential kernel functions, that works for both exciting and inhibiting interactions, as modelled by \cite{Bremaud1996, Chen2017}.
This work builds upon the methodology for the univariate case, presented in \cite{bonnet2021}, by focusing in the intervals where the intensity function is positive.
We show that, under a weak assumption on the kernel functions, these intervals can be determined exactly.
We can then write for each dimension the integral of the intensity function, known in the literature as the compensator, which in turn provides an explicit expression of the log-likelihood.
This enables to build the corresponding maximum likelihood estimator and we complete our procedure with a variable selection step to recover the significant interactions within the whole process. This is of particular interest since it provides a graphical interpretation of the model and it can also be used a reduction dimension tool.
Our numerical procedure is implemented
in Python and freely available on GitHub.\footnote{\url{https://github.com/migmtz/multivariate-hawkes-inhibition}}
As a by-product of our method, the closed-form expression of the compensator also allows to assess goodness-of-fit via the Time Change Theorem and multiple testing.
We carry out a numerical study on simulated data and on a neuronal activity dataset \citep{Peterson2016,Radosevic2019}.
The performance of our approach is compared to estimations obtained via approximations from \cite{Bompaire2020} and \cite{Lemonnier2014}, and we show that our method not only achieves better estimations but is capable of identifying correctly the interaction network of the process.

To outline this paper, Section~\ref{sec:presentation} presents the multivariate Hawkes process framework and reviews the literature regarding inference of non-linear Hawkes processes.
In Section~\ref{sec:exponential}, we detail our procedure, including the maximum likelihood estimation, variable selection and goodness-of-fit test to assess the quality of the estimations. We also address the question of identifiability of the model, that we prove under a few sufficient conditions. 
The whole procedure is illustrated on simulated data in Section~\ref{sec:numerical} and applied to a neuronal activity dataset in Section~\ref{sec:neuron}. In Section \ref{sec:discussion} we discuss our contributions and its current limitations along with interesting perspectives for future work.

\section{The multivariate Hawkes process}\label{sec:presentation}
    \subsection{Definition}
    A multivariate Hawkes process $N = (N^1, N^2, \ldots, N^d)$ of dimension $d$ is defined by $d$ point processes on $\RR_+^*$, denoted $N^i\colon \mathcal B(\RR_+^*) \to \NN$, where $\mathcal{B}(\RR_+^*)$ is the Borel algebra on the set of positive numbers.

    Each process $N^i$ can be characterised by its associated event times $\left(T_k^i\right)_k$ and its conditional intensity function, defined for all \(t \ge 0\) by
    \begin{align}
    \lambda^i(t) &= \adaptedpar{\mu_i + \sum_{j=1}^{d}{\int_{0}^{t}{h_{ij}(t-s)\,\mathrm{d}N^j(s)}}}^+\nonumber\\ &= \adaptedpar{\mu_i + \sum_{j=1}^{d}\sum_{T_k^j \leq t}{h_{ij}(t-T^j_k)}}^+\label{eq:dimensional_c_intensity}\,,
    \end{align}
    where
    $x^+ = \max{(0,x)}$.
    Here, the quantity $\mu_i\in\RR_+^*$ is called the baseline intensity and each interaction or kernel function $h_{ij}\colon \RR_+^* \to \RR$ represents the influence of the process $N^j$ on the process $N^i$ and $T^j_k$ corresponds to the $k$-th event time of process $N^j$.

    \begin{remark}
      The positive-part function in Equation~\eqref{eq:dimensional_c_intensity} is needed to ensure the non-negativity of \(\lambda^i\) in the presence of strong inhibiting effects, that is when some interaction functions \(h_{ij}\) are sufficiently negative compared to positive contributions.
      Concretely, the positive part does not affect the intensity function if inhibiting effects are in minority compared to the positive contributions (exciting effects or baseline intensities).
    \end{remark}
    \begin{remark}
        Equation~\eqref{eq:dimensional_c_intensity} may question the reader for two reasons.
        First, it is the cadlag definition of a the conditional intensity of a Hawkes process.
        It's our choice to prefer it to the caglad version but all the results presented here can be written in this setting.
        Second, it is considered that the history is empty for \(t < 0\).
        It is a common choice for statistical inference (a finite amount a times is observed) while the infinite history is preferred for a probabilistic analysis based on a stationary assumption.
    \end{remark}

    For each process \(N^i\) and for all $t\geq 0$, let us note $N^i(t) = \sum_{k\geq 1}{\II_{T_k^i \leq t}}$ the measure of \((0, t]\) and the compensator \[\Lambda^i(t) = \int_{0}^{t}{\lambda^i(u)\,\mathrm{d}u}\,.\]

    The process $N$ can be seen as a point process on $\RR_+^*$, where for any $B\in\mathcal{B}(\RR^*_+)$, $N(B) = \sum_{i=1}^{d}{N^i(B)}$.
    Similarly to a univariate process, \(N\) can be characterised by its conditional intensity \(\lambda\) (also called total intensity):
    \begin{equation}\label{eq:total_intensity}
        \lambda(t) = \sum_{i=1}^{d}{\lambda^i(t)}\,,
    \end{equation}
    and by its compensator
    \[\Lambda(t)=\int_{0}^{t}{\lambda(u)\,\mathrm{d}u} =\sum_{i=1}^{d}{\Lambda^i(t)}\,.\]

    From this point of view, the process \(N\) is associated to event times $\left(T\park\right)_k = \left(T_{u_k}^{m_k}\right)_k$, corresponding to the ordered sequence composed of $\bigcup_{i=1}^{d}\{T_k^i \mid k>0\}$,
    and we may define, for every \(t\ge 0\), $N(t) = \sum_{k\geq 1}{\II_{T\park \leq t}} = \sum_{i=1}^{d}{N^i(t)}$.
    Here, $(u_k)_k$ is the random ordering sequence and $(m_k)_k$ the sequence of marks that make it possible to identify to which dimension each time corresponds.
    These marks can be written as \[m_k = \sum_{j=1}^{d}{j\II_{N^j\left(\{T\park\}\right) = 1}}\,.\]
    
    \begin{remark}
    A more detailed introduction of multivariate point processes via the concept of marked point processes can be found in \cite[Chapter 6.4]{DaleyV1}
    \end{remark}

    As the aim of this paper is to describe a practical methodology for estimating the conditional intensities \(\lambda_1\), \dots, \(\lambda_d\) via maximising the log-likelihood, the latter quantity has to be made explicit.
    Let \(t \ge 0\) and assume that event times
    $\left\{ T_k^i : 1 \le k \le N_i(t), 1 \le i \le d \right\}$
    are observed in the interval \((0, t]\).
    Then, given a parametric model $\mathcal{P} = \{(\lambda_{\theta_1}^1, \dots, \lambda_{\theta_d}^d)\colon \theta = (\theta_1, \dots, \theta_d) \in \Theta_1 \times \dots \times \Theta_d\}$ (and associated compensators \(\Lambda_{\theta_1}^1, \dots, \Lambda_{\theta_d}^d\)) for conditional intensity functions \(\lambda^1, \dots, \lambda^d\), for every \(\theta \in \Theta\), the log-likelihood $\ell_t(\theta)$ reads \cite[Proposition 7.3.III.]{DaleyV1}
    \begin{equation*}
      \ell_t(\theta) = \sum_{i=1}^{d}{\ell_t^i(\theta_i)}\,,\nonumber 
    \end{equation*}
    with
    \begin{equation}\label{eq:general_log_likelihood}
      \ell_t^{i}(\theta_i) = \sum_{k=1}^{N^i(t)}{\log{\lambda^i_{\theta_i}(T_k^{i-})}} - \Lambda_{\theta_i}^i(t)\,,
    \end{equation}
    where $\lambda^i_{\theta_i}(T_k^{i-}) = \lim_{t\to T_k^{i-}}{\lambda_{\theta_i}^i(t)}$ and with convention $\log{(x)} = -\infty$ for $x\leq 0$.

    The heart of the problem in deriving a maximum likelihood estimator for the conditional intensities \(\lambda^i\) is being able to evaluate exactly the compensator values \(\Lambda_\theta^i(t)\) for every possible \(\theta \in \Theta\), which requires to determine when the conditional intensities \(\lambda^i\) are non-zero.
    The forthcoming sections clear this point up.

  \subsection{Related work}\label{sec:literature}
    Estimation methods for Hawkes processes have focused mainly on self-exciting interactions (by assuming $h_{ij} \geq 0$). In \cite{Ozaki1979}, the author presents the maximum likelihood estimation method for univariate processes with exponential kernel \(h(t) = \alpha \mathrm{e}^{-\beta t}\) (\(\alpha>0\), \(\beta>0\)), the same method being established in \cite{Mishra2016} for the power law kernel function \(h(t) = \frac{\alpha\beta} {(1 + \beta t)^{1 + \gamma}}\) (\(\alpha>0\), \(\beta>0\), \(\gamma > 0\)). In \cite{ChenHawkes} the maximum likelihood method is presented for the multivariate version with exponential kernel, while \cite{Bompaire2020} proposed an inference method based on optimising a least-squares criterion. Other methods in the parametric setting include spectral analysis \citep{Adamopoulos1976}, EM algorithm \citep{Veen2006} and method of moments \citep{DaFonseca2013}.

    Estimators of the interaction functions \(h_{ij}\) are also presented in a nonparametric setting. For instance, \cite{Yang2017} proposed to estimate \(h_{ij}\) in a reproducing kernel Hilbert space.
    \cite{Reynaud2014} proposed a decomposition of the interaction functions \(h_{ij}\) on a histogram basis with bounded support, the estimation of which are obtained by minimising a least-squares contrast.
    Hawkes processes with excitation have also been studied in a Bayesian context, with likelihood-based approaches, as in \cite{Rasmussen2013} for the univariate case and in \cite{Donnet2020} for multivariate processes.

    Although inhibiting effects in Hawkes processes were first mentioned in \cite{Bremaud1996}, they have only met a growing interest in the last decade.
    Concerning inference, most of the known methods are not designed for handling the inhibiting case: nevertheless some are able in practice to estimate negative interactions by minimising a least-squares criterion, but without guaranteeing that the estimated intensity functions remain non-negative \citep{Reynaud2014, Bompaire2020}.
    A similar approach is proposed in \cite{Lemonnier2014} for maximum likelihood estimation, where the compensator \(\Lambda^i\) is approximated by integrating
    the conditional intensity $\lambda^i$ without the positive part function (see Equation~\eqref{eq:dimensional_c_intensity}).
    Obviously, these methods should perform well when the intensities remain mostly positive, but it is unclear how they will adapt to scenarios when the intensities are frequently equal to zero due to inhibiting terms.
    A similar remark is mentioned in \cite{Bacry2016} concerning their estimation method, that can provide negative estimations of the interactions only if there is a negligible chance of the intensities to be null.
    
    Inference procedures that are specifically dedicated to Hawkes processes with inhibition are scarcer in the literature. \cite{Sulem2021} presents various results for non-linear Hawkes processes including inhibition effects: existence, stability and Bayesian estimation for kernel functions with bounded support.
    \cite{Deutsch2022} presents choices of priors for Bayesian estimation based on a new reparametrisation of the process.

    Lastly, \cite{bonnet2021} presents a maximum likelihood estimation adapted to the univariate Hawkes process with inhibition and monotone kernel functions.
    The decisive contribution of this work is to give, for an exponential kernel \(h(t) = \alpha \mathrm{e}^{-\beta t}\) (\(\alpha \in \RR\), \(\beta>0\)), a closed-form expression of restart times, which are defined as the instants at which the single conditional intensity becomes non-zero.
    This makes possible to compute explicitly the compensator and then the log-likelihood.
    Yet, this study is limited to the univariate case.
    It has to be noted that a formalism similar to \cite{bonnet2021} for multivariate Hawkes processes is mentioned in \cite{Deutsch2022},
    but used neither for maximum likelihood estimation, nor for goodness-of-fit tests.

    This paper goes a step forward in estimation of multivariate Hawkes processes with inhibition, by providing the first exact maximum likelihood method for exponential interactions $h_{ij}(t) = \alpha_{ij}\mathrm{e}^{-\beta_{ij}t}$, combined with a variable selection procedure.
    As it will be explained in the next section, the proposed approach also enables to perform standard goodness-of-fit tests.

    \section{Estimation and goodness-of-fit}\label{sec:exponential}

    \subsection{Introductive example}
    Before motivating and explaining the estimation procedure proposed in this paper, we present an example of multivariate Hawkes process.
    Figure~\ref{fig:2_dimension_example} depicts in red conditional intensities \(\lambda^1\) and \(\lambda^2\) for a realisation of a \(2\)-dimensional Hawkes process (see the forthcoming section for the definition of underlying intensities). The existence of such a process (along with its stationarity) is ensured by controlling the spectral radius $\rho(S^+) < 1$ of the matrix $S^+ = (\|h_{ij}^+\|_1)_{ij}$ \citep{Deutsch2022}. Similar results with slightly different conditions can be found in \cite{Bremaud1996} and \cite{Sulem2021}.
    The simulation has been carried out with baselines \(\mu_1 = 0.5\) and \(\mu_2 = 1.0\),
    and with exponential kernels \(h_{ij}(t) = \alpha_{ij} \mathrm{e}^{-\beta_{ij}t}\) parameterised by:
    \[
      \begin{pmatrix}
      \alpha_{11} & \alpha_{12}\\
      \alpha_{21} & \alpha_{22}
      \end{pmatrix}=
      \begin{pmatrix}
      -1.9 & 3.0\\
      0.9 & -0.7
      \end{pmatrix},
      \qquad
      \text{and}
      \qquad
      \begin{pmatrix}
      \beta_{11} & \beta_{12}\\
      \beta_{21} & \beta_{22}
      \end{pmatrix} =
      \begin{pmatrix}
      2.0 & 20.0\\
      3.0 & 2.0
      \end{pmatrix}\,.
    \]
    These kernels have been chosen such that both processes are self-inhibiting ($\alpha_{11}, \alpha_{22} < 0$) but inter-exciting ($\alpha_{12}, \alpha_{21} > 0$).

    \begin{figure*}[!ht]
    \centering
    \includegraphics[width=.75\linewidth]{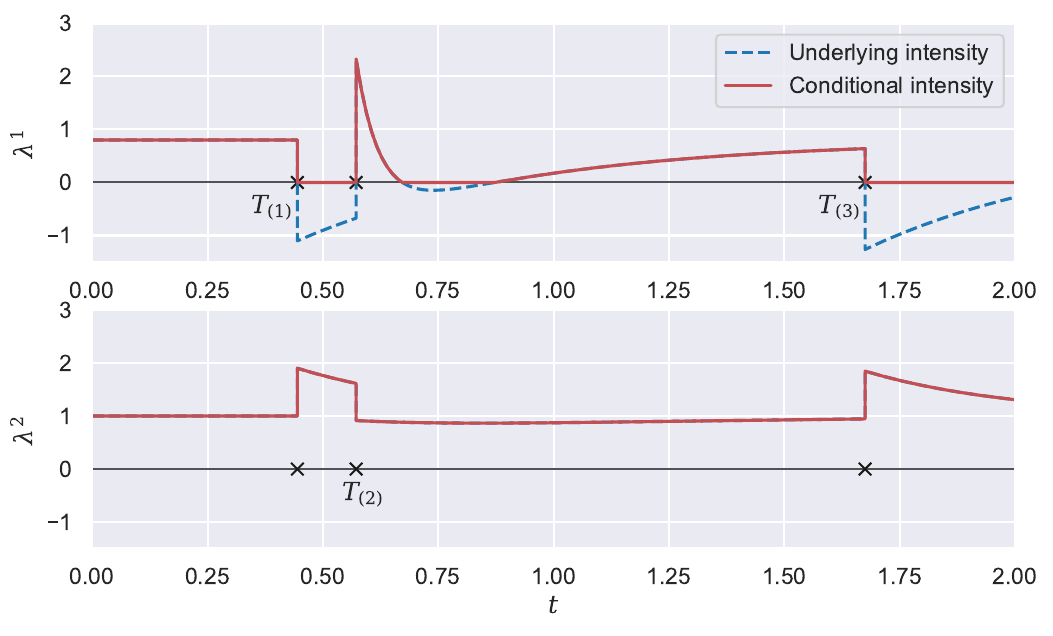}
    \caption{Simulation of a \(2\)-dimensional Hawkes process. Each cross corresponds to an event time, and each $T\park$ is shown in its corresponding process.}
    \label{fig:2_dimension_example}
    \end{figure*}

    The goal of this paper is to establish a parametric estimation method, via maximum likelihood estimation, that is able to handle both excitation and inhibition frameworks in the multivariate case.
    For this purpose, it is necessary to compute explicitly the log-likelihood \(\ell_t(\theta)\) (see Equation~\eqref{eq:general_log_likelihood}) and in particular to evaluate the compensator \(\Lambda_\theta^i\), expressed as an integral of \(\lambda_\theta^i\).
    For the latter, the main challenge is to determine when conditional intensities \(\lambda^i\) are non-zero, that is on which intervals they are tailored by the exponential interaction functions and not by the positive-part operator.

    In \cite{bonnet2021}, the authors solved this challenge for univariate processes by remarking that the conditional intensity is monotone between two event times.
    Figure~\ref{fig:2_dimension_example} illustrates that this is not necessarily true for multivariate processes (here, between \(T_{(2)}\) and \(T_{(3)}\)).
    This constitutes the major difficulty we have to cope with.

  \subsection{Underlying intensity and restart times in the multivariate setting}

    From now on, let us focus on the exponential model \citep{Hawkes1971}, where each interaction function $h_{ij}$ is then defined as
    \[h_{ij}(t) = \alpha_{ij}\mathrm{e}^{-\beta_{ij}t}\,,\]
    with $\alpha_{ij}\in\RR$ and $\beta_{ij}\in\RR_+^*$ for $i,j\in \{1,\ldots, d\}$.
    For each $i\in\{1,\ldots, d\}$, the underlying intensity function $\lambda^{i\star}$ is defined as in \cite{bonnet2021} for the univariate case:
    \[\lambda^{i\star}(t) = \mu_i + \sum_{j=1}^{d}{\int_{0}^{t}{h_{ij}(t-s)\,\mathrm{d}N^j(s)}}\,.\]
    This quantity coincides with the conditional intensity \(\lambda^i\) when it is non-zero, and is non-positive otherwise.
    In particular, we can observe that \(\lambda^i (t) = \left( \lambda^{i\star}(t) \right)^+\) (see Figure~\ref{fig:2_dimension_example}).

    As explained in the previous section, the main difficulty of the multivariate exponential setting is the non-monotony of conditional intensities \(\lambda^i\) between two event times.
    Determining intervals where \(\lambda^i\) is non-zero (that is when \(\lambda^{i\star}\) is positive) would require to numerically find the roots of a high-degree polynomial, which is expensive and inexact.
    To alleviate this problem, we introduce Assumption~\ref{assu:beta}.

    \begin{assumption}\label{assu:beta}
    For each $i\in\{1,\ldots, d\}$, there exists $\beta_i\in\RR_+^*$ such that $\beta_{ij} = \beta_i$ for all $j\in\{1,\ldots, d\}$.
    \end{assumption}

    \begin{remark} This model with constant recovery rates $\beta_i$ has been studied before in the works of \cite{Ogata1981} in the self-exciting version of a 2-dimensional Hawkes process. Intuitively, this assumption considers the situation where the rate of ``dissipation'' of any internal or external effect is dependent only on the receiving phenomenon. For instance, for neuronal interactions, each activation from neuron $j$ will have an impact on a connected neuron $i$ dependent on both neurons $(\alpha_{ij})_{ij}$ but the ``recovery'' time can be assumed to depend only on the connected neuron $i$ \((\beta_{i})_i\).
    \end{remark}

    As we will see in Lemma~\ref{lemma:restart_times}, this assumption enables to recover the monotony of the conditional intensities between two times.
    It remains now to determine when the underlying intensity \(\lambda^{i\star}\) is negative.
    To do so, we define the restart times in the multivariate framework, to be,
    for any \(k\) and \(i\):
    \[T\park^{i\star} = \min\adaptedpar{\inf{\{t\geq T\park \colon \lambda^{i\star}(t) \geq 0\}}, T\park[k+1]}.\]
    
    {\begin{figure*}[!ht]
    \centering
    \includegraphics[width=0.8\linewidth]{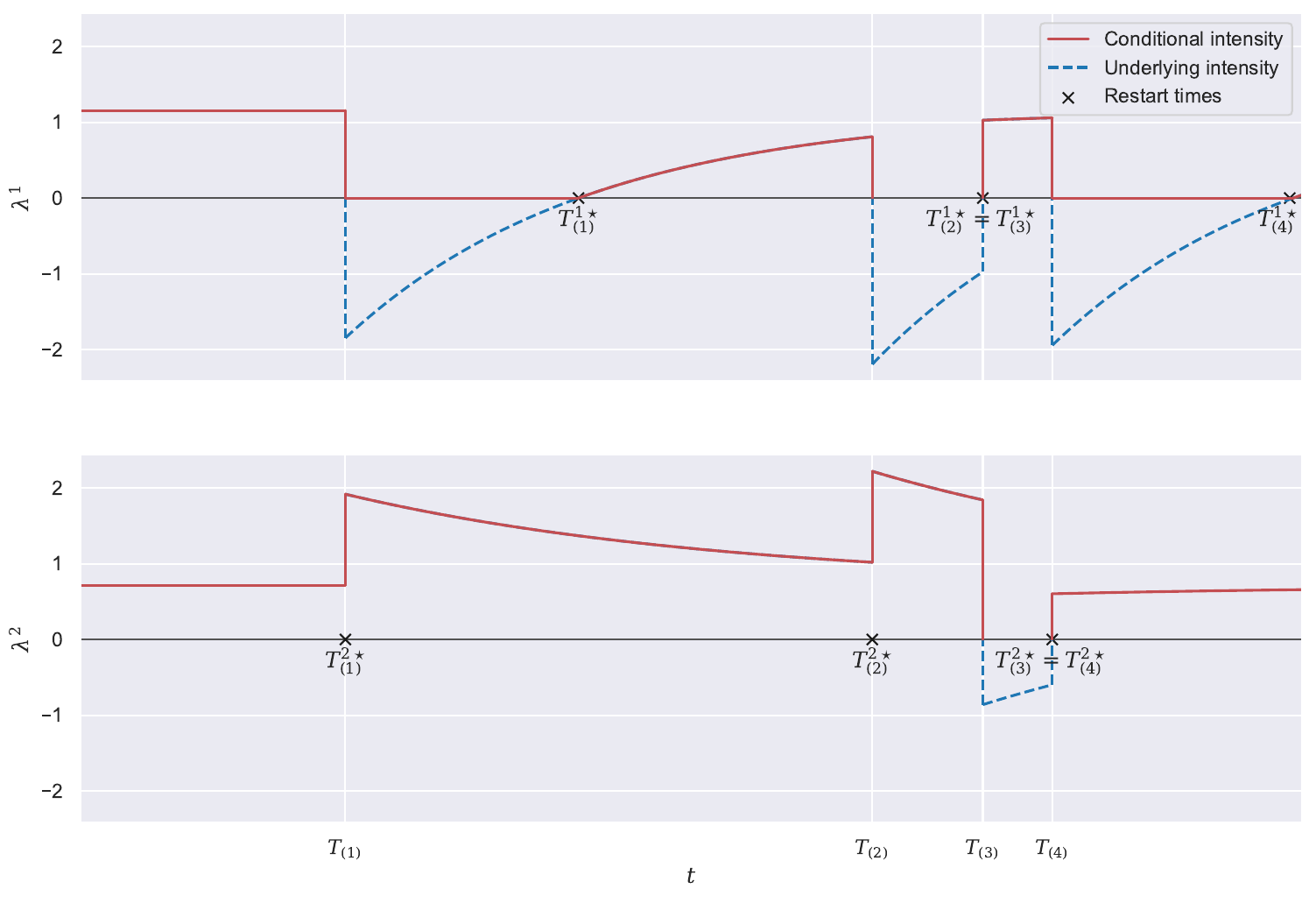}
    \caption{Illustration of restart times $(T\park^{\star})_{1\leq k \leq 4}$ for each subprocess of a 2-dimensional process associated with event times $(T\park[k])_{1\leq k \leq 4}$.}
    \label{fig:restarttimes}
    \end{figure*}}
    
     As exemplified in Figure~\ref{fig:restarttimes}, the restart time $T\park^{i\star}$ (associated to the sub-process $i$) corresponds to the first instant between $T\park$ and $T\park[k+1]$ from which $\lambda^{i\star}(t)$ becomes non-negative (or $T\park[k+1]$ if this instant does not exist).
    Intuitively, it means that $\lambda^{i}(t) = \lambda^{i\star}(t)$ on $(T\park^{i\star}, T\park[k+1])$ and $\lambda^{i}(t) = 0$ elsewhere on  $(T\park$, $T\park[k+1])$.
    This is formalised in Lemma~\ref{lemma:restart_times}.
    In particular, it appears that the restart time $T\park^{i\star}$ can be expressed as a function of $$T\park + \beta_i^{-1}\log{\adaptedpar{\frac{\mu_i - \lambda^{i\star}(T\park)}{\mu_i}}}\,,$$
      which is the single root to the equation $\mu_i + (\lambda^{i\star}(T\park) - \mu_i)\mathrm{e}^{-\beta_i(t-T\park)} = 0$ on the interval $[T\park, +\infty)$ when $\lambda^{i\star}(T\park) < 0$ (see top panel of Figure~\ref{fig:restarttimes}).
      Then, Proposition~\ref{prop:compensator} gives an explicit formulation of the compensator of each subprocessi \(N^i\), which is needed to compute the log-likelihood (see Equation~\eqref{eq:general_log_likelihood}).
      The proofs of these two results are presented respectively in Appendix~\ref{appendix:proof_lemma} and Appendix~\ref{appendix:proof_prop_compensator}.

    \begin{lemma}\label{lemma:restart_times}
    If Assumption~\ref{assu:beta} is granted, then for each $i\in\{1,\ldots, d\}$ and any $k\geq1$: 

    \begin{equation}\label{eq:restart_times_exponential}
        T\park^{i\star} = \min{\adaptedpar{t^\star_k, T\park[k+1]}}\,.
    \end{equation}
    where
    \[t^\star_k = \adaptedpar{T\park + \beta_i^{-1}\log{\adaptedpar{\frac{\mu_i - \lambda^{i\star}(T\park)}{\mu_i}}}\II_{\{\lambda^{i\star}(T\park) < 0\}}}\]

    Furthermore,
    for all \(t \in (T\park, T\park[k+1])\),
    \[
      \lambda^i(t) =
      \begin{cases}
        \lambda^{i\star}(t) > 0 & \text{if } t \in (T\park^{i\star}, T\park[k+1])\\
        0 & \text{otherwise}.
      \end{cases}
    \]
    \end{lemma}

    \begin{proposition}\label{prop:compensator}[Compensator for multivariate exponential kernels]
    Let us suppose that Assumption~\ref{assu:beta} is granted. For each $i\in\{1,\ldots, d\}$ the compensator $\Lambda^i$ of the process $N^i$ reads, $\forall t\geq 0$:
    \begin{equation}\label{eq:compensator_exponential}
      \Lambda^i(t) =
      \begin{cases}
        \mu_i t & \text{if } t < T\park[1] \\
        \mu_i T\park[1] + \sum_{k=1}^{N(t)} J_k & \text{if } t \geq T\park[1] \,,
      \end{cases}
    \end{equation}
    where for all integer \(k \in \{1, \dots, N(t)\}\):
    \begin{equation*}
      J_k =
      \mu_i \left[\min(t, T\park[k+1] ) - T\park^{i\star}\right] + \beta_i^{-1} \left(\lambda^{i\star}(T\park) - \mu_i \right)
      \left[ \mathrm{e}^{-\beta_i(T\park^{i\star} - T\park)} - \mathrm{e}^{-\beta_i(\min(t, T\park[k+1])- T\park)} \right]\,.
    \end{equation*}
    \end{proposition}
    
  \subsection{Identifiability and likelihood computation}\label{sec:likelihood}
    As already mentioned in Section~\ref{sec:presentation}, let \(t \ge 0\) and assume that event times
    $\left\{ T_k^i : 1 \le k \le N_i(t),\right.$ $\left. 1 \le i \le d \right\}$ are observed in the interval \((0, t]\).
    We consider the parametric exponential model for a multivariate Hawkes process of dimension $d$, defined by
    \[
      \mathcal P = \left\{(\lambda^1, \dots, \lambda^d)\colon \lambda^1 \in \mathcal P^1, \dots, \lambda^d \in \mathcal P^d \right\} \,,
    \]
    where for each \(i \in \{1, \dots, d\}\), \(\mathcal P^i\) is the exponential parametric model for the process \(N^i\):
\[
      \mathcal{P}^i
      = \left\{	\lambda_{\theta_i}^i = \left(\mu_i + \sum_{j=1}^{d}{\int_{-\infty}^{t}{\alpha_{ij} \mathrm{e}^{-\beta_i(t-s)}\,\mathrm{d}N^j(s)}}\right)^+ \colon \theta_i = (\mu_i, \alpha_{i1}, \ldots, \alpha_{id}, \beta_i) \in\Theta \,\right\}\,,
    \]
    where $\Theta = \RR_+^\star\times\RR^d\times\RR_+^\star$.
    For a candidate set of intensities \((\lambda_{\theta_1}^1, \dots, \lambda_{\theta_d}^d)\), the underlying intensity functions are denoted \(\lambda_{\theta_i}^{i\star}\) (\(i \in \{1, \dots, d\}\)), the compensators $\Lambda_{\theta_i}^i$
    and the restart times
    $(T^{i\star}_{\theta_i, (k)})_{k}$.
    Now, given a realisation $\left(T\park \right)_{k>0}$ of a multivariate exponential Hawkes process,
    Theorem~\ref{theorem:identifiability} establishes the correspondence between the conditional intensities and the parameters.
        
\begin{theorem}[Identifiability]\label{theorem:identifiability}
    Let $N=(N^1, \ldots, N^d)$ be a multivariate Hawkes process defined by a set of intensity functions $(\lambda^1_{\theta_1}, \dots, \lambda^d_{\theta_d}) \in \mathcal P$, for some $\theta = (\theta_1, \dots, \theta_d) \in \Theta^d$.
    Let also $\left(T\park \right)_{k>0}$ be a realisation of $N$ and $\mathcal{F}_t$ be the corresponding filtration.
    
    Let us assume that a.s. for every $(i, j) \in \{1, \dots, d\}^2$, $i\neq j$,
    there exist an event time \(\tau\) from process \(N^j\),
    and an event time \(\tau_+ > \tau\) from process \(N^i\), such that:
    \begin{enumerate}
        \item \label{hyp:ii} $\lambda_{\theta_i}^i(\tau^-) > 0$;
        \item \label{hyp:iii} there are only events of process $N^j$ in the interval $[\tau, \tau_+)$.
    \end{enumerate} 
    
    Then, for any $\theta' \in \Theta^d$,
    \[
      \forall i \in \{1, \dots, d\},~
      \lambda_{\theta_i}^i(t ~|~ \mathcal{F}_t) = \lambda_{\theta_i'}^i(t ~| \mathcal{F}_t) \text{ a.e.}
      \qquad
      \iff
      \qquad
      \theta = \theta' \,.
    \]
\end{theorem}

The proof is presented in Appendix \ref{appendix:proof}. To the best of our knowledge, the only identifiability result for non-linear multivariate Hawkes processes is given by \cite{Sulem2021} but only if interaction functions $h_{ij}$ have a bounded support. Their proof strongly relies on this assumption since they extend to the multivariate case the renewal properties proved by \cite{Costa2020} for non-linear univariate Hawkes processes with a bounded kernel.
Their proof also requires some assumptions to ensure that one process is not totally inhibited, which is also a consequence of the assumptions that we propose, as discussed in Section~\ref{sec:identifiability}. 

    As expected, Proposition~\ref{prop:compensator} makes it possible to compute explicitly the log-likelihood expressed in Equation~\eqref{eq:general_log_likelihood} for multivariate exponential Hawkes processes.
    This is formalised in Corollary~\ref{cor:mle} (and proved in Appendix~\ref{appendix:proof_cor_mle}).

    \begin{corollary}\label{cor:mle}
        Let \(i \in \{1, \dots, d\}\) and \(k\geq 2\) an integer.
        Let us denote
        \[
          S_k^i := T\park[N(T_k^i)-1]
          = T\park[\max\{\ell \in \mathbb N^* : {T\park[\ell]} < T_k^i\}] \,,
        \]
        the time preceding directly \(T_k^i\), the \(k^{\text{th}}\) observation of process \(N^i\).
        
        Then, for all \(\theta \in \Theta\),
        the log-likelihood of process \(N^i\) reads:
        \begin{align}\label{eq:log_likelihood}
            \ell^i_t(\theta_i) = &\log{\mu_i} + \sum_{k=2}^{N^i(t)}{\log{\adaptedpar{\mu_i + (\lambda^{i\star}_{\theta_i}(S_k^i) - \mu_i)\mathrm{e}^{-\beta_i(T_k^i - S_k^i)}}}} - \Lambda^i_{\theta_i}(t)\,,
        \end{align} where $\Lambda^i_{\theta_i}$ is given by Equation~\eqref{eq:compensator_exponential} and with convention $\log{(x)} = -\infty$ for $x\leq 0$.

    \end{corollary}

    Algorithm~\ref{alg:likelihood} in Appendix~\ref{app:log-lik} presents the iterative computation of the likelihood using Equation~\eqref{eq:log_likelihood}. In particular, the complexity of the computation is $O(N(t)\times d)$.
    It is then possible to establish the Maximum Likelihood Estimator, which we will refer to as (MLE). These estimators will be denoted by a tilde: $(\tilde\mu_i)_i$, $(\tilde\alpha_{ij})_{ij}$, $(\tilde\beta_i)_i$ and $(\tilde h_{ij})_{ij}$.

\subsection{On identifiability conditions}\label{sec:identifiability}

        In the previous section we presented a result on the identifiability of multivariate Hawkes process with inhibition via Theorem~\ref{theorem:identifiability}.
        Let us mention that identifiability of parameters \(\mu_i\) and \(\beta_i\) do not require any assumption.
	  The challenge of the proof lies in identifying parameters \(\alpha_{ij}\), which is achieved thanks to Conditions~\ref{hyp:ii} and \ref{hyp:iii} of Theorem~\ref{theorem:identifiability}.
	  Condition~\ref{hyp:ii} allows to control strong inhibition scenarios by ensuring that each subprocess' intensity is positive sufficiently often, and not only at its own event times,
	  while Condition~\ref{hyp:iii} enables to disentangle the contributions of each subprocess.
	  We believe that these assumptions are only sufficient and could be weakened at the cost of a more intricate analysis.

    In this section we will discuss this set of conditions by providing two examples of parameters $\alpha_{ij}$ that allow to apply this result along with one counter-example.

\subsubsection{Examples for which the conditions are fulfilled}
    Let us begin with an example of a two-dimensional process.
    In this situation, as soon as both processes have an infinite amount of events,
    Conditions~\ref{hyp:ii} and \ref{hyp:iii} boil down to finding indexes \(k\ge1\) and \(k'\ge1\) such that \(\lambda^1({T_k^2}^-)>0\) and \(\lambda^2({T_{k'}^1}^-)>0\).
    Indeed, it is then enough to consider \((\tau, \tau_+) = (T_k^2, T_{N(T_k^2)+1}^1)\) for \((i, j)=(1, 2)\) and \((\tau, \tau_+) = (T_{k'}^1, T_{N(T_{k'}^1)+1}^2)\) for \((i, j)=(2, 1)\).

    \begin{example}\label{ex:2dim_identifiability}
        Let us assume that $N$ is a two-dimensional Hawkes process with the following matrix for parameters $\alpha_{ij}$:
        \[\begin{pmatrix}
                      \alpha_{11} & \alpha_{12}\\
                      \alpha_{21} & \alpha_{22}
                      \end{pmatrix}=
                      \begin{pmatrix}
                      \alpha_{11} & 0.0\\
                      \alpha_{21} & \alpha_{22}
                      \end{pmatrix}\,,\]
        such as $\alpha_{11} \leq 0$, $\alpha_{21} \geq 0$ and $\alpha_{22} > 0$.

        Both processes contain an infinite number of events as process $N^1$ can be seen as a one-dimensional Hawkes process and the second one has a lower-bounded intensity $\lambda^2 > \mu^2$.
        Now, since \(\lambda^2(t)>0\) for all \(t\ge0\), we have \(\lambda^2({T_1^1}^-)>0\).
        Then, let us first remark that the event times of process $N^1$ occur independently of $N^2$, as $\alpha_{12} = 0$.
        So, for
        any event time $T_{\ell}^1$ of $N^1$, the restart times can be written as if $N^1$ was a univariate Hawkes process (see \cite{bonnet2021}):
\[T\park[\ell]^{1\star} = T_\ell^1 + \beta_1^{-1}\log\left(  \frac{\mu_1 - \lambda^{1\star}(T_\ell^1)}{\mu_1}  \right)\II_{\{\lambda^{1\star}(T_\ell^1) < 0\}}\,.\]

        For $t > T\park[\ell]^{1\star}$ small enough, both $\lambda^1(t)$ and $\lambda^2(t)$ are positive,
        meaning that the next event time can come either from \(N^1\) or from \(N^2\).
        If we consider an infinite sequence of event times, we will eventually observe an event \(T_k^2\) of $N^2$ such that $\lambda^1({T_k^2}^-)>0$.
       
    \end{example}

    This gives a set of Hawkes processes with inhibition that verify the assumptions of Theorem~\ref{theorem:identifiability}. For higher dimensions, the multiplicity of all possible connections between processes complicates the study of general cases from a theoretical point of view. Example~\ref{example:ddim_example} illustrates a case where the conditions are fulfilled by considering identically distributed processes.

    \begin{example}\label{example:ddim_example}
        Let us consider a $d$-dimensional Hawkes process,
        as well as $\mu, \alpha_+, \alpha_-, \beta$ such that for any $i$ and for any $j\neq i$:
        \begin{align*}
    	\mu_i = \mu > 0\,,     \quad    \alpha_{ii} = \alpha_- \leq 0\,,\\ 
	    \beta_i = \beta > 0\,, \quad \alpha_{ij} = \alpha_+ \geq 0\,.
        \end{align*}
        As each process has the same parameters for $\mu_i$ and $\beta_i$ along with the exact same interactions, all processes are identically distributed and so in order to verify Conditions~\ref{hyp:ii} and \ref{hyp:iii} it is enough to verify them for $i=1$ and $j=2$.

        Fulfilling both conditions amounts to verifying that, with non-zero probability, we can find $k\ge1$ such that:
        \begin{enumerate}
            \item $\lambda^1({T_k^2}^-) > 0$;
            \item $T\park[N(T_k^2) + 1]$ is an event from process $N^1$.
        \end{enumerate}

        Furthermore, let us notice that as all processes are cross-exciting, if $\lambda^1({T_k^2}^-) > 0$ for a certain $k$, then $\lambda^1(t) > 0$ for $t>T_k^2$ small enough, so $T\park[N(T_k^2) + 1]$ can come from process $N^1$. It remains now to verify that with non zero probability $\lambda^1({T_k^2}^-) > 0$ for a certain $k\ge1$. 

        But the opposite would entail that almost surely, for all indexes \(k\ge1\), $\lambda^1({T_k^2}^-) = 0$.
        Yet, as all processes are identically distributed,
        this means that, for all $j\neq 2$, \(\lambda^j({T_k^2}^-) = 0\) almost surely.
        It would then follow that, for every $i$ and for every $j\neq i$, for all $k> 1$, $\lambda^j({T_k^i}^-) = 0$ almost surely. Let us consider $T_k^{i_0}$ for a fixed $k$ and $i_0$. As each process $N^j$ for $j\neq i$ is then excited with the same parameter $\alpha_+$ and they are identically distributed, then $\lambda^j(T_k^{i_0})$ are identically distributed and independent conditionally on history $\mathcal{F}_{T_k^{i_0}}$. It follows then that the restart times associated to $T_k^{i_0}$ of each process are independent and identically distributed. Consequently, there is a non-zero probability that at least two processes regenerate at roughly the same time before another time of process $N^{i_0}$ (as it self-inhibits). Once that for $j_0, j_1$, $\lambda^{j_0}$ and $\lambda^{j_1}$, are positive, the next event time may come from either process, and so either $\lambda^{j_1}(T\park[N(T_k^{i_0}) + 1]^-) > 0$ with $T\park[N(T_k^{i_0}) + 1]$ from process $N^{j_0}$, or the inverse, which contradicts the fact that for all $i,j$ and for all $k>1$, $\lambda^j(T_k^i) = 0$.
        So for at least one pair $(i, j)$ and one $k$, we must have $\lambda^{j}({T_k^{i}}^-) > 0$ and so it has to occur for all pairs, in particular $(1, 2)$.

    \end{example}

    In the next section we present Example~\ref{ex:2dim_counterexample} of a specific set of parameters for which both conditions are not necessarily met.

\subsubsection{Example where the conditions are not fulfilled}
	\begin{example}\label{ex:2dim_counterexample}
    Let us consider the Hawkes process defined by the following parameters:
		\[
            	     \begin{pmatrix}
                      \mu_{1} \\
                      \mu_{2} 
                      \end{pmatrix} =
                      \begin{pmatrix}
                      10^5 \\ 
                      0.1
                      \end{pmatrix}\,,
                      \quad
                      \begin{pmatrix}
                      \alpha_{11} & \alpha_{12}\\
                      \alpha_{21} & \alpha_{22}
                      \end{pmatrix}=
                      \begin{pmatrix}
                      1.0 & 1.0\\
                      -10^6 & -10^6
                      \end{pmatrix}\,,
                      \quad
                      \begin{pmatrix}
                      \beta_{1} \\
                      \beta_{2} 
                      \end{pmatrix} =
                      \begin{pmatrix}
                      1.0\\ 
                      10^{-5}
                      \end{pmatrix}\,.
    		\]
            The probability that the first event time $T\park[1]$ is from process $N^1$ is $10^5 / (10^5 + 0.1) \approx 1$. If the first event is indeed from process $N^1$, then process $N^2$ is strongly inhibited and in that case the restart time $T_{\park[1]}^{2\star}$ is equal to \[T_{\park[1]}^{2\star} = T_1 + 10^5 \log(10^7)\,,\]
            which is very far from \(T_1\).
            As $\lambda^1$ is lower-bounded by $10^5$, the next candidate of process $N^1$ is roughly distributed as an exponential random variable with parameter $10^5$ so the next event time is with high probability of process $N^1$.
            If this is the case, process $N^2$ is further inhibited,
            and with probability going exponentially quickly to \(1\), all next event times will come from process \(N^1\),
            preventing us from granting Conditions~\ref{hyp:ii} and \ref{hyp:iii}.

	\end{example}

 \subsection{Recovering the graph of interactions}\label{sec:support}
   The aim of this section is to describe methodologies able to estimate non-null interactions between subprocesses, which boils down to detecting parameters such that \(\alpha_{ij} \neq 0\).
   Recovering interactions has an interest, first, in providing a graphical interpretation of the Hawkes model, as it describes which subprocesses are actually connected within the whole process.
   Moreover, the graph of interactions can also be used as a reduction dimension tool,
   for instance if we focus on the dynamic of one single subprocess,
   the activity of which can be impacted by a sub-network of surrounding processes that we want to identify,
   as investigated in \cite{neuro2022} for neuronal activity.
   
   Estimating non-null interactions is a challenging topic, which has been particularly studied for linear regression (see for instance \cite{tibshirani1996}).
   Regarding Hawkes processes with inhibition, this is even more demanding because of the non-differentiability of the log-likelihood.
   Therefore, we propose, in the subsequent sections, two post hoc techniques (i.e.\ after computing the MLE estimator) not related to numerical optimization, and additionally, having the benefit to scale easily to high-dimensional processes.
      
   \subsubsection{Thresholding}
   
   The first method that will be referred to as (MLE-$\varepsilon$) is obtained by adding a thresholding step to the classic Maximum Likelihood Estimation (MLE). This is similar to the cumulative percentage of total variation approach presented in Principal Component Analysis \cite[Section 6.1.1]{Joliffe2002}. All absolute estimated values $\lvert \tilde\alpha_{ij}\rvert$ are arranged in increasing order $(\lvert \tilde\alpha_{(k)}\rvert)_{k\in \{1, \ldots, d^2\}}$. We compute then the cumulative sums
   $s_k = \sum_{l=1}^{k}{\lvert \tilde\alpha_{(l)}\rvert}$
   and write $S := s_{d^2}$ the sum of all absolute estimated values.
   Then all estimations $\tilde\alpha_{(k)}$ such that \[s_k < \varepsilon S\,,\] are set to zero, for a threshold $\varepsilon \in (0,1)$.
  Subsequently, all non-null estimations $\tilde\alpha_{ij}$ are then re-estimated by maximising the log-likelihood.
   
   The choice of an optimal threshold level $\varepsilon$ requires a way of comparing estimations, which is achieved thanks to the goodness-of-fit procedure described in Section \ref{sec:goodness},
   and which is illustrated in Section \ref{sec:description_methods}.
   
   \subsubsection{Confidence interval}\label{sec:confidence_interval}
   
   The second method is applicable when a sample $(N_1, \ldots, N_n)$ of $n$ realisations of a multivariate Hawkes process is available. For every $i,j$, we average all estimations $\tilde\alpha_{ij}$ over the realisations $N_1, \ldots, N_n$ to obtain $\bar\alpha_{ij}$ and then we determine a confidence interval around each estimation at a given confidence level $1-\eta$. Then each estimation for which the confidence interval contains $0$ is set to zero. Subsequently, all non-null estimations $\tilde\alpha_{ij}$ are re-estimated.
   
   In this paper we consider two different confidence intervals. 
   \begin{itemize}
   \item (CfE) corresponds to the empirical interval
   \[
      \left[\alpha_{(\lfloor{\frac{\eta}{2} n}\rfloor)}, \alpha_{(\lceil (1-\frac{\eta}{2}) n \rceil)}\right]\,,
    \]
    where, $(\alpha_{(k)})_{k\in \{1, \ldots, n\}}$ is the sequence of the sorted estimations of $\alpha_{ij}$,
    and $\lfloor \cdot \rfloor$ and $\lceil \cdot \rceil$ are respectively the floor and the ceiling functions.
   \item (CfSt) corresponds to
   \[\left[\bar\alpha_{ij} - t_{1-\frac{\eta}{2}} s_n, \bar\alpha_{ij} + t_{1-\frac{\eta}{2}} s_n \right]\,,\]

   where $s_n$ is the empirical standard deviation of the sample and $t_{1-\frac{\eta}{2}}$ is the quantile of level $1-\frac{\eta}{2}$ of the Student distribution with $n-1$ degrees of freedom. 
   This corresponds to a confidence interval obtained for normally distributed estimators.
   For Hawkes processes with exclusively exciting interactions, estimations obtained through the MLE procedure are asymptotically normal as proven in \cite[Theorem 5]{Ogata1978} and as discussed in \cite{Laub2014}.
   For processes with inhibiting interactions, 
   asymptotic normality is still an open question but is in all likelihood true.
   However, one has to be careful when using this estimator, in particular a small number of observations could imply that the asymptotic normality is not achieved. In practice, normality can be tested thanks to a Kolmogorov-Smirnov test.   
   
   \end{itemize}
   
   This method of selection through confidence intervals can be seen as testing the following hypothesis at a confidence level $1-\eta$ for every $i,j$:
   
   \[
   \begin{dcases}
       \mathcal{H}_0 : \alpha_{ij} = 0 \,,\\
       \mathcal{H}_1 : \alpha_{ij} \neq 0\,.
   \end{dcases}
   \]
   We can then compute the corresponding $p$-value for each test and set to zero all parameters for which the null hypothesis is not rejected. As we test $d^2$ different hypotheses, it is essential to incorporate multiple testing procedures. For this purpose, we choose the Benjamini-Hochberg method, consisting in adapting the rejection threshold of each $p$-value. This method enables to control the false discovery rate (FDR). If we denote $V$ the number of rejected true null hypothesis and $R$ the number of rejected true alternative hypotheses, the FDR is defined as \[FDR = \EE\left[\frac{V}{R+V}\right]\,.\] In other words, we control the expected number of true null hypotheses (i.e. parameter $\alpha_{ij}$ is equal to zero) rejected by our testing method. The B-H procedure considers the ordered $p$-values $(p_{(k)})_{k\in\{1, \ldots, d^2\}}$ and compares each one to the adapted rejection threshold $(1-\eta) \frac{k}{d^2}$. Then, we determine the largest $K\in\{1, \ldots, d^2\}$ such that $p_{(K)} < (1-\eta) \frac{K}{d^2}$ and we reject all hypothesis such that $p_{(k)} \leq p_{(K)}$

  \subsection{Goodness-of-fit}\label{sec:goodness}
    As a benefit of our approach, it is possible to perform a goodness-of-fit test for assessing the quality of estimations. This is particularly useful when choosing between several estimations (such as those introduced before), in particular to choose an optimal level of thresholding for the (MLE-$\varepsilon$) method.
    The closed-form expression of the compensator given in Proposition~\ref{prop:compensator} enables the use of the Time Change Theorem for inhomogeneous Poisson processes \cite[Proposition 7.4.IV]{DaleyV1}.
    For any $i$, the sequence of transformed times $(\Lambda^i(T_k^i))_k$ is a realisation of a homogeneous Poisson process with unit-intensity if and only if $(T_k^i)_k$ is a realisation of a point process with intensity $\lambda^i$.

    We can then define for any $\theta\in\Theta$ the null hypothesis
    \[
      \mathcal H_i : \text{``} (T_k^i)_k \text{ is a realisation of a point process with intensity } \lambda_{\theta_i}^i \text{''}.
    \]
    
    The hypothesis is then tested via a Kolmogorov-Smirnov test between the empirical distribution $(\Lambda_\theta^i(T_{k+1}^i) - \Lambda_\theta^i(T_k^i))_{k\geq 1}$ and an exponential distribution with parameter $1$. We obtain then $d$ different tests with $p$-values $(p_i)_{i\geq 1}$.
    Using multiple testing approaches can help in determining correctly estimated processes.

    Lastly, we can obtain an additional test by considering the entire sequence of times $(T\park)_{k\geq 1}$ and the total intensity $\lambda$. We obtain then the null hypothesis
    \[
      \mathcal{H}_{tot} : \text{``} (T\park)_k \text{ is a realisation of a point process with intensity } \lambda_\theta \text{''},
    \]
    with corresponding $p$-value $p_{tot}$. This value is obtained by way of a Kolmogorov-Smirnov test between the empirical distribution of $(\Lambda_\theta(T\park[k+1]) - \Lambda_\theta(T\park[k]))_{k\geq 1}$ and an exponential distribution with parameter $1$, this time using the total compensator of the process.

    In the forthcoming sections, this testing procedure is applied to several realisations of event times, that are independent of the considered estimator.
    This enables to assess properly the accuracy of estimations, without knowing the true conditional intensities.
    This is particularly interesting for real-world data.

    Let us mention that the goodness-of-fit procedure is not only an assessment of the overall fit between the model and the observations but it provides also a tool to calibrate the threshold for the (MLE-$\varepsilon$) method. Indeed, for each threshold level $\varepsilon$ chosen over a grid, we compute all p-values (one for each subprocess, one for $p_{tot}$) and we choose the value of $\varepsilon$ that maximises the mean p-value. For the other two methods, (CfE) and (CfSt), the model selection procedure is described in Section~\ref{sec:confidence_interval} and the goodness-of-fit is only performed afterwards in order to establish the quality of the estimations.

\section{Illustration on synthetic datasets}\label{sec:numerical}
    \subsection{Simulation procedure}
    In order to assess the performance of the maximum likelihood estimation method, we simulate different data by using Ogata's thinning method \citep{Ogata1981}. This method consists in defining a piecewise constant function $\lambda^+$ such that for any $k\geq1$ and any $t\in[T\park, T\park[k+1])$, $\lambda^+(t) \geq \lambda(t)$. For this, we define $\lambda^+$ for any $t\in[T\park, T\park[k+1])$ as \[\lambda^+(t) = \sum_{i=1}^{d}{\left(\mu_i + \sum_{j=1}^{d}\int_{0}^{t}{\alpha_{ij}^+\mathrm{e}^{-\beta_i(T\park-s)}\,\mathrm{d}N^j(s)}\right)}\,,\] which corresponds to considering only the positive interactions.

    Four different parameter sets are considered: three sets for 2-dimensional Hawkes processes and a last one for a 10-dimensional process.
    Table~\ref{tab:simulated_2_table} presents the parameters used in Dimension 2. All scenarios contain at least one negative interaction ($\alpha_{ij}<0$). Scenario (1) is a Hawkes process where all parameters are non-null whereas Scenarios (2) and (3) are chosen to study the performance of our methods when estimating null interactions ($\alpha_{12}$ for Scenario (2) and $\alpha_{21}$ for Scenario (3)).
    All simulations have exactly \(5000\) event times in total.
    \begin{table}[h!]
    \begin{center}  
        \centering
        \begin{tabular}{c|c|c|c}

             Scenario & (1) & (2) & (3) \\
             \toprule
             $\begin{pmatrix}
            \mu^1\\
            \mu^2
            \end{pmatrix}$ & $\begin{pmatrix}
            0.5\\
            1.0
            \end{pmatrix}$ & $\begin{pmatrix}
            0.7\\
            1.0
            \end{pmatrix}$&$\begin{pmatrix}
            1.2\\
            1.0
            \end{pmatrix}$\\\midrule
            $\begin{pmatrix}
            \alpha_{11} & \alpha_{12}\\
            \alpha_{21} & \alpha_{22}
            \end{pmatrix}$&$\begin{pmatrix}
            -1.9 & 3.0\\
            1.2 & 1.5
            \end{pmatrix}$&$\begin{pmatrix}
            0.2 & 0.0\\
            -0.6 & 1.2
            \end{pmatrix}$&$\begin{pmatrix}
            -1.0 & 0.1\\
            0.0 & -0.8
            \end{pmatrix}$\\\midrule
            $\begin{pmatrix}
            \beta_1\\
            \beta_2
            \end{pmatrix}$ & $\begin{pmatrix}
            5.0\\
            8.0
            \end{pmatrix}$ & $\begin{pmatrix}
            3.0\\
            2.0
            \end{pmatrix}$&$\begin{pmatrix}
            0.3\\
            0.5
            \end{pmatrix}$
        \end{tabular}
        \caption{Parameters for simulations of two-dimensional Hawkes processes.}
        \label{tab:simulated_2_table}
    \end{center}
    \end{table}

    In order to carry out the hypothesis testing procedure, we simulate a sample of Hawkes processes independent from the one used for the estimation. Each testing sample contains as many realisations as the estimation sample. All $p$-values presented in the paper correspond to the average obtained over all realisations.

    \subsection{Proposed methods and comparison to existing procedures}
    \label{sec:description_methods}
    The main focus of this paper is to assess the performance of the maximum likelihood estimator to correctly detect the interacting functions of our processes without ambiguity, estimators are denoted with a tilde: $(\tilde\mu_i)_i$, $(\tilde\alpha_{ij})_{ij}$, $(\tilde\beta_i)_i$ and \((\tilde h_{ij})_{ij}\).
    In this paper we propose four methods previously introduced in Section \ref{sec:support}:
    \begin{itemize}
        \item (MLE) The estimator obtained by minimising the opposite of the log-likelihood $-\sum_{i=1}^{d}{\ell_t^i(\theta)}$ (see Equation~\eqref{eq:log_likelihood}). The log-likelihood is computed via Algorithm~\ref{alg:likelihood} and the minimisation is done with the L-BFGS-B method \citep{Byrd1995}.
        \item (MLE-$\varepsilon$) The estimator obtained by adding a thresholding step to the previous method to determine the non-null estimations.  The value of $\varepsilon$ is chosen such that it maximises the mean over all $p$-values obtained.
        \item (CfE) The estimator whose support is obtained through the empirical confidence intervals.
        \item (CfSt) The estimator using Student distributed intervals after verification of normality of the estimations.
    \end{itemize}

    The latter three methods are specially interesting for Scenarios (2) and (3) of the 2-dimensional processes, and also for the 10-dimensional setting.

    \begin{remark}
    Another option considered for (MLE-$\varepsilon$) is to use instead the values $\lvert\tilde \alpha_{ij} / \tilde \beta_{i}\rvert$ for the thresholding. Numerical results slightly differ between the two methods, with the retained method showing better overall $p$-values.
    \end{remark}

    For an informative assessment of the proposed approaches, we compare their performance to estimation methods from the literature.
    However, up to our knowledge, there is no other parametric estimation methods designed for inhibiting processes, that is for handling negative values of \(\alpha_{ij}\).
    As a consequence, we chose to include estimation methods developed for exciting processes, that are nonetheless able to produce negative estimations of \(\alpha_{ij}\).
    This is the case for three popular approaches described below.

    \begin{enumerate}
        \item (Approx) The first one \citep{Lemonnier2014} is obtained by approaching the compensator $\Lambda^i(t)$ (in each log-likelihood $\ell^i_t(\theta)$) by \[\int_{0}^{t}{\lambda^{i\star}(u)\,\mathrm{d}u}\,.\]

        In the case where all interactions are positive, this integral is equal to the compensator. The difference is when interactions are negative as this integral takes into account the negative values of the underlying intensity function. The minimisation is done in the same way as for (MLE) using the L-BFGS-B method.

        \item The other two methods minimise the least-squares loss approximation defined in \cite{Reynaud2014,Bompaire2020} as:
        \[R_t(\theta) =  \int_{0}^{t}{(\lambda_{\theta}(u))^2\,\mathrm{d}u} - \frac{2}{t}\sum_{k=1}^{N(t)}{\lambda^{m_k}_{\theta}(T\park^-)} \,,\]
        which is an observable approximation of \(\|\lambda_\theta - \lambda\|_t^2 = \int_0^t (\lambda_\theta(u) - \lambda(u))^2 \, \mathrm du\) up to a constant term.
        In \cite{Bompaire2020}, all interactions are assumed to be positive, however the implemented version of this method in the package \texttt{tick} \cite{Bompaire2018} allows to retrieve negative values. For this, we consider two different kernel functions from this implementation:
        \begin{itemize}
        \item (Lst-sq) $h_{ij}(t) = \alpha_{ij}\beta_{ij}e^{-\beta_{ij}t}$, where $\beta_{ij}$ is fixed beforehand by the practitioner. In practice, we fix $\beta_{ij} = \beta_i$ to be consistent with our model (see Assumption~\ref{assu:beta}). The only solver in the implementation that provides negative values is BFGS,
        which is limited to work with an $\ell_2$-penalty. The grid of values $\{1, 10, \ldots, 10^6\}$ is considered for the regularisation constant. To obtain the best estimation for this method, we choose the constant that minimises the relative squared error over all estimated parameters. 
        
        \item (Grid-lst-sq) $h_{ij}(t) = \sum_{u=1}^{U}{\alpha_{ij}^u\beta^u e^{-\beta^u t}}$, with $(\beta^u)_u$ a fixed grid of parameters. In our case, we choose $U=d$ and the grid contains each parameter $\beta_{i}$. Intuitively, by applying an $\ell_1$-penalty, this method would be able to retrieve the corresponding parameter $\beta_{i}$ for each process.
        However, in practice, the implementation uses BFGS as optimiser and is limited to work with an $\ell_2$-penalty. As for (Lst-sq), the regularisation parameter is chosen over the grid of values by minimising the relative squared error.
    \end{itemize}
    \end{enumerate}

    \subsubsection{Results on bivariate Hawkes processes}\label{sec:dim2}

   We generate 25 realisations for each parameter set given in Table~\ref{tab:simulated_2_table} and we estimate the parameters for each individual simulation.  We begin this comparison by competing the proposed methods with both (Approx) and (Lst-sq) which are the two methods with the same kernel functions considered in this paper. Figure~\ref{fig:boxplots} displays the relative squared errors for each group of parameters (baselines $(\mu_i)_i$, interaction terms $(\alpha_{i,j})_{i,j}$ and delay factors $(\beta_i)_i$)  by considering vector norms.
    
    {\begin{figure*}[!ht]
    \centering
    \includegraphics[width=0.9\linewidth]{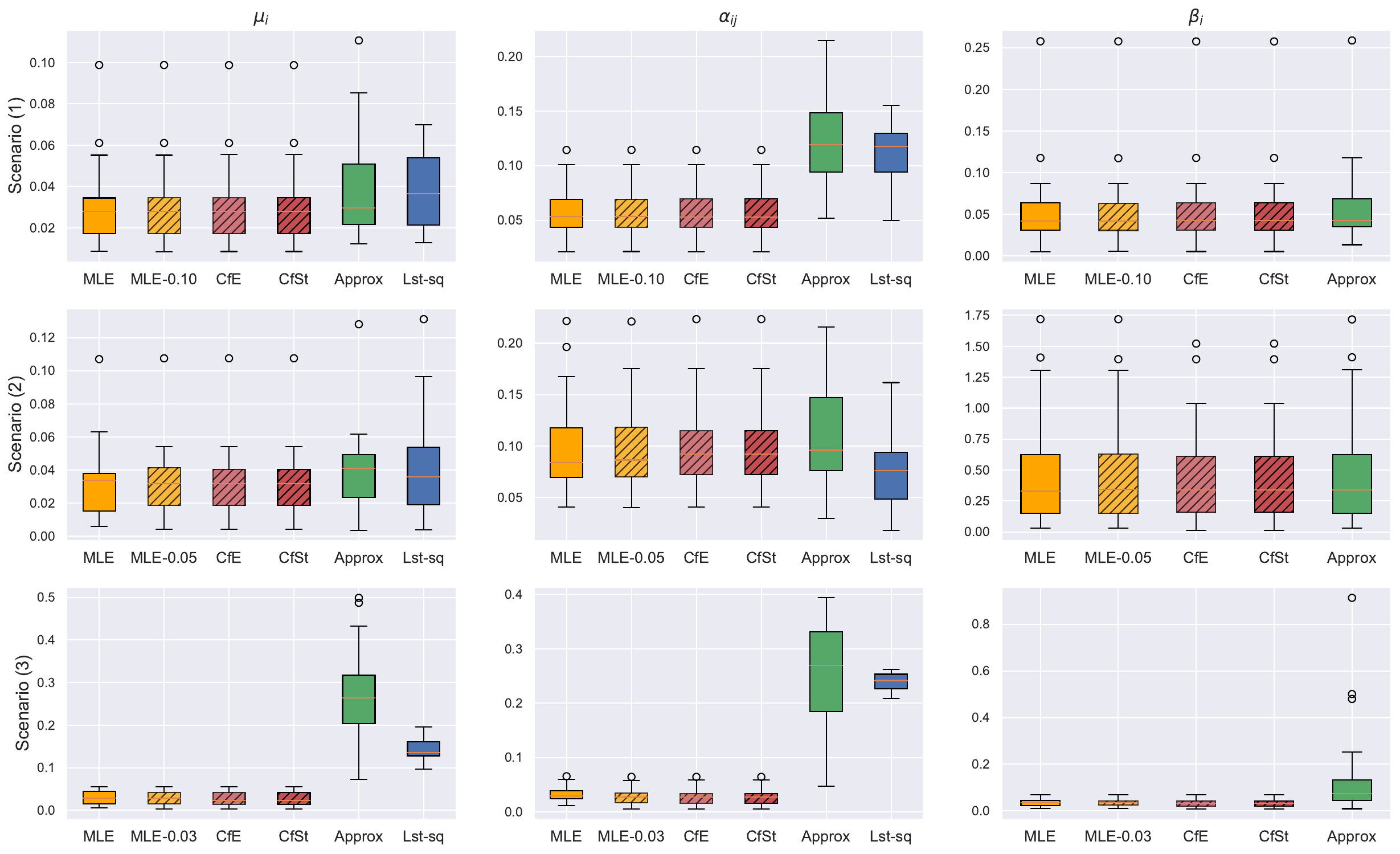}
    \caption{Boxplots of the relative squared error for each group of parameters ($(\mu_i)_i$, $(\alpha_{i,j})_{i,j}$ and $(\beta_i)_i$) for 25 realisations of a two-dimensional Hawkes processes.
    (Lst-sq) does not appear in the last column because it is provided with the true values of \((\beta_i)_i\).
    The proposed methods are (MLE), (MLE-$\varepsilon$), (CfE) and (CfSt).}
    \label{fig:boxplots}
    \end{figure*}}
    
    First, we observe that delay factors \((\beta_i)_i\)
    (last column of Figure~\ref{fig:boxplots})
    are similarly estimated by all approaches.
    Let us recall that (Lst-sq) is not included in the comparison of delay factors:
    since it requires to provide a value for these parameters (they are not estimated),
    it was given the true values of \((\beta_i)_i\) as input.
    An alternative offered by \texttt{tick} is to provide a grid
    of values, but this approach, denoted (Grid-lst-sq), is included in the comparison at the end of the section because of its difference with the exponential model considered here.
    
    Then, regarding the baseline intensities \((\mu_i)_i\) and the interaction factors \((\alpha_{ij})_{ij}\),
    the proposed methods outperform the two other approaches.
    In all Scenarios, (MLE-$\varepsilon$), (CfE) and (CfSt) appear to perform almost identically as they retrieve the same supports and from then, the re-estimations are the same.
    In Scenario~(2), all estimation methods perform reasonably well. 
    This can be explained by the weak inhibiting effect of the interaction $1\to2$, leaving the intensity almost always positive.
    The slight difference between (MLE-$\varepsilon$) and the confidence intervals comes from the fact that (MLE-$\varepsilon$) is applied individually to each estimation so for some estimations it does not set any values to zero.

    In Scenario~(1), the performance of (Approx) and (Lst-sq) is altered, in particular for the \((\tilde \alpha_{ij})_{ij}\) estimations, because the inhibiting effect is stronger than in Scenario~(2).
    The major changes appear in Scenario~(3), where both (Approx) and (Lst-sq) obtain very high relative errors.
    More precisely, they fail to explain the interactions between the two processes (see the estimations \((\tilde \alpha_{ij})_{ij}\) in the middle column of Figure~\ref{fig:boxplots}), which is compensated by a wrong estimation \((\tilde \mu_i)_i\) of baseline intensities.
    This is not surprising since Scenario~(1), and even more Scenario~(3), were designed so that the intensity functions are frequently equal to zero, which induces major differences between true and underlying intensities.
    Since (Approx) and (Lst-sq) are both based on assuming that these two functions are almost equal, the violation of this assumption causes large estimation errors.
    As expected, the proposed methods, which are developed to handle such inhibiting scenarios, provide accurate estimations.
    
    These results are confirmed by the outcomes of the goodness-of-fit test displayed in Table~\ref{tab:p_values_2}.
    It shows indeed the averaged $p$-values for each scenario using both the true parameters and all four estimations from Figure~\ref{fig:boxplots} with 25 simulations different from the ones used for estimation. In particular, we can see that our methods obtain high $p$-values, being very close to those obtained using the true parameters.
    Table~\ref{tab:p_values_2} also highlights when parameters are incorrectly estimated.
    For instance, in Scenario (1), (Approx) correctly estimate Process~2 but provides less accurate estimations for Process~1 (the \(p\)-value is almost half the one obtained with the true parameters), which is the one characterised by a self-inhibiting behaviour.
    In addition, at least one of the proposed methods obtains the highest value for $p_{tot}$ in each scenario, which illustrates the ability of these procedures to reconstruct the complete process $N$.
    Let us note that the very low \(p\)-values obtained by (Approx) and (Lst-sq) for Scenario (3) confirm the ability of the goodness-of-fit procedure to detect when the parameter estimations strongly differ from the true parameters.
    
    \begin{table*}[!ht]
    \begin{center}
    \centering
    \begin{tabular}{c|ccc|ccc|ccc}
          & \multicolumn{3}{c|}{Scenario (1)}& \multicolumn{3}{c|}{Scenario (2)}& \multicolumn{3}{c}{Scenario (3)} \\
         $p$-value & $p_1$ & $p_2$ & $p_{tot}$ & $p_1$ & $p_2$ & $p_{tot}$ & $p_1$ & $p_2$ & $p_{tot}$\\
         \toprule
         True & 0.492 & 0.438 & 0.430 & 0.535 & 0.468 & 0.479 & 0.510 & 0.623 & 0.338\\
         \midrule
         MLE &  0.440 & 0.442  & 0.398  & 0.483  & 0.461  & 0.485 & 0.549  & 0.638 & $0.357$\\
         \midrule
         MLE-$\varepsilon$ & \multirow{3}{*}{0.440} & \multirow{3}{*}{0.442} &  \multirow{3}{*}{0.398} & \multirow{3}{*}{0.488}  & \multirow{3}{*}{0.461} & \multirow{3}{*}{0.491} & \multirow{3}{*}{0.549} & \multirow{3}{*}{0.574}  & \multirow{3}{*}{0.327}  \\
         CfE &  &  &  &  &  &  &  &  &  \\
         CfSt &  &  &  &  &  &  &  &  &  \\
         \midrule
         Approx & 0.257 & 0.442 & 0.358 & 0.483 & 0.452 & 0.459 & $\mathbf{0.0}$ & $\mathbf{0.007}$ & $\mathbf{0.0}$ \\
         Lst-sq & 0.154 & 0.438 & 0.392 & 0.534 & 0.463 & 0.478 & $\mathbf{0.0}$ & $\mathbf{0.0}$ & $\mathbf{0.0}$
    \end{tabular}
    \caption{Average $p$-values for estimations of two-dimensional Hawkes processes for all scenarios. The values are averaged over 25 simulations. In bold the $p$-values correspond to a rejected hypothesis at a confidence level of $0.95$.}
    \label{tab:p_values_2}
    \end{center}
    \end{table*}

     Lastly, let us investigate the estimations obtained via (Grid-lst-sq), which can be used in practice as a way to estimate the parameters $\beta_i$ by providing a grid of possible parameters.
     Let us mention that both of the previous comparisons (boxplots and \(p\)-values) cannot be done here due to the difference in the number of parameters, but we can compare the methods in terms of
     reconstructions $\tilde h_{ij}$ of the interaction functions $h_{ij}$.
     For this purpose, we analyse Figure~\ref{fig:reconstruction_param_1}, which represents the estimated interaction functions \(\tilde h_{ij}\) for all methods in Scenario (3).
     Interestingly, we see that (Grid-lst-sq) performs similarly to (Lst-sq),
     while the latter is fed with all true values \((\beta_i)_i\) for each interaction.
     However, we see that (Grid-lst-sq)  suffers from the same difficulties than (Approx) and (Lst-sq), which was expected since it relies on the same unvalid assumption.
     Let us note that we chose to display the results for Scenario (3) since it highlights the main differences between the compared approaches but the reconstructions for Scenarios (1) and (2)
     can be found in Appendix~\ref{app:numerical}
     (Figures \ref{fig_supp:scenario_1} and \ref{fig_supp:scenario_2}).
     
     {\begin{figure*}[!ht]
     \centering
     \includegraphics[width=0.9\linewidth]{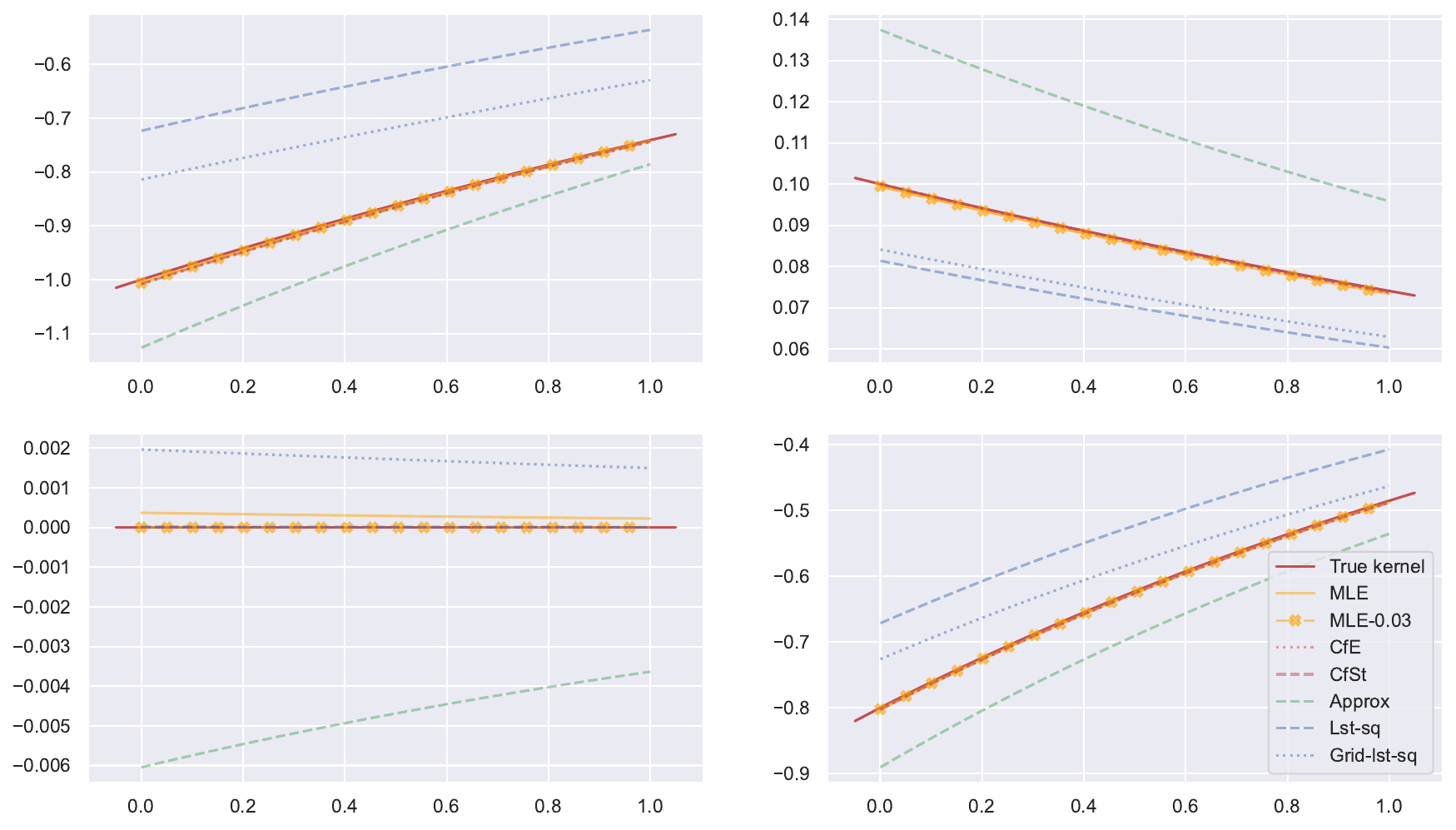}
     \caption{Reconstruction of interaction functions $h_{ij}$ for Scenario (3) of two-dimensional Hawkes processes along with all estimated functions $\tilde h_{ij}$. The real function is plotted in red and 25 estimations are averaged for each method.}
     \label{fig:reconstruction_param_1}
     \end{figure*}}

  \subsubsection{A $10$-dimensional Hawkes process}
    A \(10\)-dimensional Hawkes process is simulated based on a set of parameters corresponding to the quantities $(\sign(\alpha_{ij})\|h_{ij}\|_1)_{ij} = (\alpha_{ij}/\beta_i)_{ij}$ displayed in Figure~\ref{fig:heatmap_10}.
    The chosen parameters fulfil the existence condition $\|\rho(S^+)\| < 1$.

    {\begin{figure*}
    \centering
    \includegraphics[clip,width=0.75\linewidth]{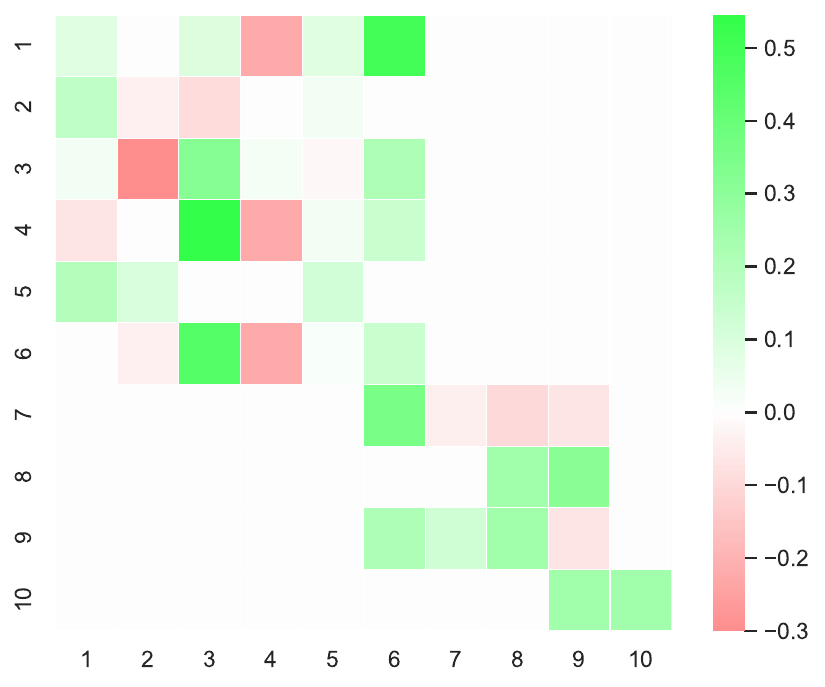}
    \caption{Heatmap of real parameters $(\sign(\alpha_{ij})\|h_{ij}\|_1)_{ij} = (\alpha_{ij}/\beta_i)_{ij}$ for the $10$-dimensional simulation.}
    \label{fig:heatmap_10}
    \end{figure*}}
    {\begin{figure*}
    \centering
    \includegraphics[clip, width=\linewidth]{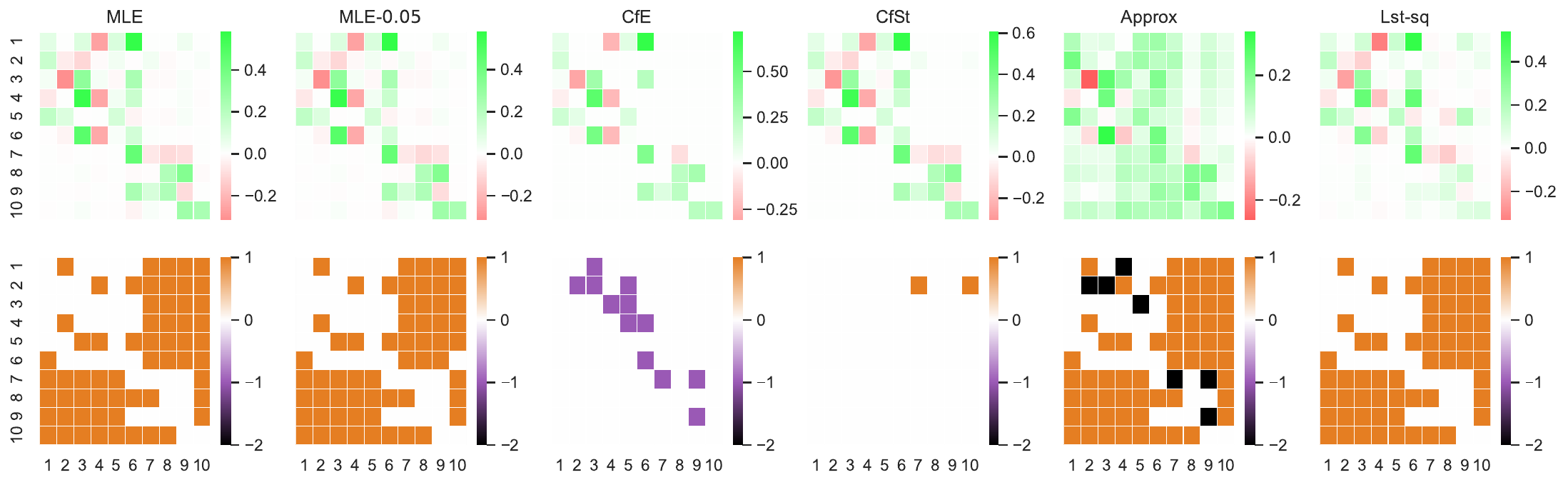}
    \caption{Top row corresponds to the heatmap for each estimation method. Bottom row corresponds to errors made with respect to real parameters from Figure \ref{fig:heatmap_10}. A value of 1 (orange) shows an undetected 0 (non-null estimation for $\alpha_{ij} = 0$), a value of -1 (purple) shows a non-null value set to 0 and a value of -2 (black) shows a non-null value whose sign is wrongly estimated. Each of the compared approaches is described in Section \ref{sec:description_methods}.}
    \label{fig:heatmap_estimated}
    \end{figure*}}

    The corresponding estimations \(\tilde \alpha_{ij}\) and \(\tilde \beta_{i}\) are averaged over 25 realisations and displayed in Figure~\ref{fig:heatmap_estimated}.
    The heatmap representation is convenient for high-dimensional processes as it allows us to see whether the signs of each interaction are well-estimated and whether the null-interactions are correctly detected.

    In this example we decided to keep only (Approx) and (Lst-sq) as comparison methods as these are the ones with the same parameterisation for the kernel functions.
    Among the four methods considered, (Approx) is the only one that wrongly estimates the sign of some interactions, represented by the black boxes in the second row matrix. (MLE) and (Lst-sq) correctly retrieve the sign of each interactions but are unable to detect the null interactions: this is not surprising since (MLE) does not contain a regularisation step and the (Lst-sq) estimator is implemented with a $\ell_2$-penalty which does not provide a sparse solution.
    On the one hand, (MLE-$\varepsilon$) is in this case quite conservative by setting a single value equal to 0 compared to (MLE).
    On the other hand, both confidence intervals methods improve the number of interactions whose sign is correctly estimated. Interestingly, we see that (CfE) sets more values equal to zero than it should (purple boxes) whereas (CfSt) denotes the opposite effect by not detecting null interactions (orange boxes). Overall, (CfSt) obtains the best results in terms of support recovery and sign estimations by committing only two errors. Table~\ref{tab:p_values_10} summarises the $p$-values for each hypothesis as described in Section~\ref{sec:goodness}. 
    All of the proposed methods obtain overall better $p$-values with no particularly low values, which is not the case for (Approx) (see $p_4$ and $p_{tot}$) and for (Lst-sq) (see $p_8$ and $p_{10}$). Although the $p$-values all exceed $5\%$, they remain substantially smaller than those obtained with the alternative methods.  
    
    \begin{table*}[!ht]
    \begin{center}
    \setlength{\tabcolsep}{5pt}
    \centering
    \begin{tabular}{c|cccccccccc|c}
         $p$-value & $p_1$ & $p_2$ & $p_3$ & $p_4$ & $p_5$ & $p_6$ & $p_7$ & $p_8$& $p_9$ & $p_{10}$ & $p_{tot}$\\
         \toprule
         True & 0.437 & 0.608 & 0.435 & 0.517 & 0.534 & 0.45 & 0.43 & 0.47 & 0.533 & 0.509 & 0.464\\
         \midrule
         MLE & 0.451 & 0.619 & 0.399 & 0.466 & 0.506 & 0.464 & 0.424 & 0.386 & 0.45 & 0.483 & 0.434\\
         \midrule
         MLE-0.05 & 0.454 & 0.622 & 0.392 & 0.466 & 0.505 & 0.427 & 0.418 & 0.392 & 0.495 & 0.482 & 0.432 \\
         CfE & 0.424 & 0.532 & 0.393 & 0.528 & 0.532 & 0.301 & 0.444 & 0.427 & 0.516 & 0.505 & 0.439\\
         CfSt & 0.452 & 0.633 & 0.375 & 0.474 & 0.527 & 0.462 & 0.431 & 0.422 & 0.488 & 0.493 & 0.465\\
         \midrule
         Approx & 0.411 & 0.376 & 0.475 & 0.077 & 0.485 & 0.411 & 0.3 & 0.384 & 0.285 & 0.436 & 0.085\\
         Lst-sq & 0.422 & 0.63 & 0.344 & 0.456 & 0.416 & 0.439 & 0.411 & 0.096 & 0.579 & 0.157 & 0.423\
    \end{tabular}    
    \caption{$p$-values for estimations of a ten-dimensional Hawkes process. The values are averaged over 25 simulations. $p_{tot}$ corresponds to testing whether the estimated intensity function corresponds to a multivariate Hawkes process $N$ as defined in Section \ref{sec:goodness}.}
    \label{tab:p_values_10}
    \end{center}
    \end{table*}

Finally, we compare the relative squared errors for each group of parameters (see Figure \ref{fig:errors_10_dim}). Similarly to the two-dimensional case, all proposed methods perform significantly better than alternative approaches regarding the estimation of all parameters.
In addition, it can be noticed that inaccurate estimations of the parameters $\alpha_{ij}$ tend to deteriorate the estimations of $\mu_i$, which suggests an effect of compensation between these parameters.
Regarding the estimator (CfSt), which shows the best averaged performance, it can be remarked that it also exhibits a large variance, in particular when estimating $\beta_i$. 

Figure~\ref{fig:p_values_support} illustrates for (CfST) the ordered $p$-values for hypothesis $\mathcal{H}_0 : \alpha_{ij} = 0$.

{\begin{figure}[!ht]
     \centering
     \includegraphics[width=0.75\linewidth]{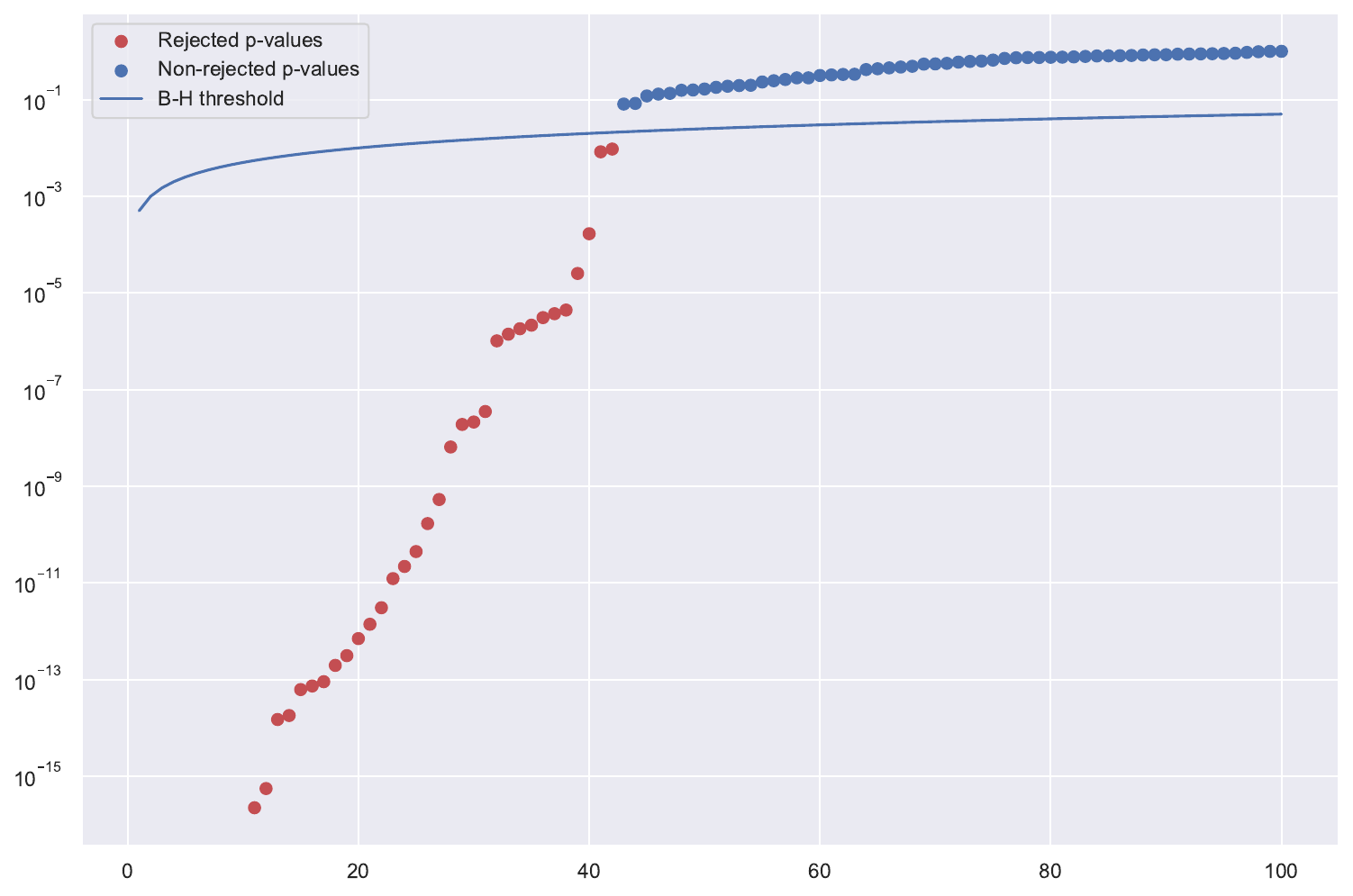}
     \caption{Ordered $p$-values corresponding to support estimation of method (CfSt). Red points correspond to rejected null hypothesis $\mathcal{H}_0 : \alpha_{ij} = 0$ and blue points to non-rejected ones.}
     \label{fig:p_values_support}
     \end{figure}}

This can be explained by this estimator providing a very sparse solution (as seen in Figure \ref{fig:heatmap_estimated}) and therefore taking into account less observations for estimating the coefficients $\beta_i$.

{\begin{figure*}[!ht]
     \centering
     \includegraphics[width=0.9\linewidth]{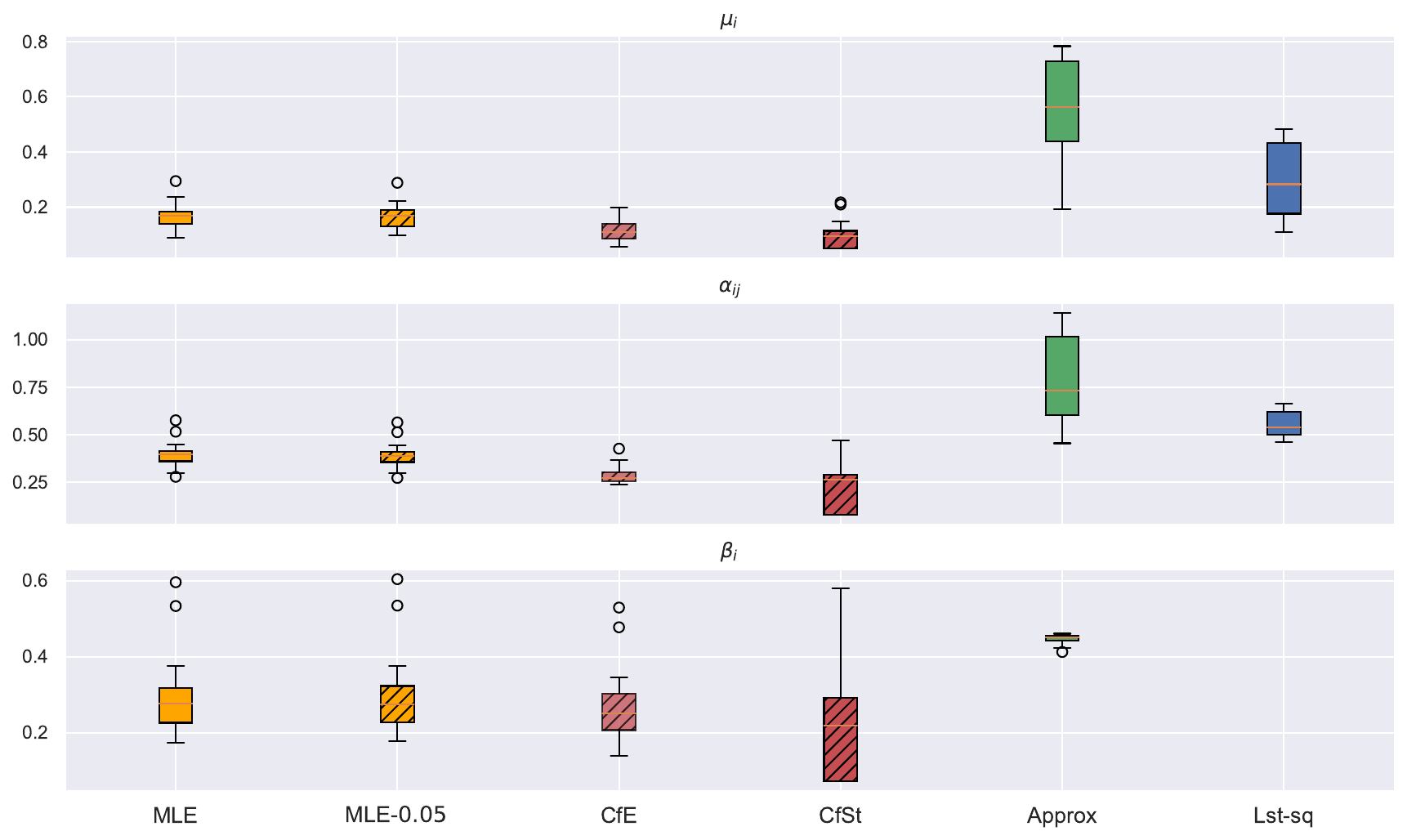}
     \caption{Boxplots of the relative squared error for each group of parameters ($(\mu_i)_i$, $(\alpha_{ij})_{i,j}$ and $(\beta_i)_i$) for 25 realisations of a ten-dimensional Hawkes processes.
    (Lst-sq) does not appear in the last row because it is provided with the true values of \((\beta_i)_i\).
    The proposed methods are (MLE), (MLE-$\varepsilon$), (CfE) and (CfSt).}
     \label{fig:errors_10_dim}
     \end{figure*}}

An important question for any inference method, especially in a high-dimensional setting, is its computational cost. Table~\ref{tab:times_10} shows the average estimation time (over 25 realisations), all times being total estimation time. More precisely, for our 3 model selection methods (MLE-$\varepsilon$), (CfE) and (CfSt), it takes into account the total times, including the first (MLE) estimation in addition to the re-estimation over the support. 

\begin{table*}[!ht] 
    \begin{center}
    \setlength{\tabcolsep}{2pt}
    \centering
    \begin{tabular}{c|c|ccc|c|cc}
          & MLE & MLE-$\varepsilon$ & CfE & CfSt & Approx & Lst-sq & Grid-lst-sq\\
         \toprule
         Computing time & 87.6 & 178.2 & 147.2 & 152.2 & 51.4 & 1.32 & 47.7
    \end{tabular}
    \caption{Average computing time in seconds for all estimation methods, averaged over 25 realisations of a ten-dimensional Hawkes process.  The time shown for (MLE-$\varepsilon$) is the total time of estimation for 7 different values of $\varepsilon$. Similarly for \texttt{tick} methods, the time is for 7 different levels of penalisation (as done in the estimations of Section~\ref{sec:dim2}). For the (Grid-lst-sq), we provided a search grid for $\beta_i$ that contains 12 values, including the true one.}
    \label{tab:times_10}
    \end{center}
    \end{table*}

Although the difference between (Lst-sq) and all other methods is substantial, let us recall that (Lst-sq) requires to be provided with parameters $\beta_i$, which offers two numerical advantages: it does not need to optimise for parameters $\beta_i$ which are the more difficult parameters to estimate and it includes a pre-computation step (of exponential terms) that accelerates all internal computations. In order to provide a fairer comparison, we include the computing time of the (Grid-lst-sq) method which can be considered as an alternative for estimating these parameters when given a search grid for $\beta_i$ (here with 12 values, including the true parameters).
As expected, the computational cost of the MLE method is higher than the alternative approaches (Approx) and (Grid-lst-sq), designed for the linear model.
It remains nevertheless that (Lst-sq) and (Grid-lst-sq) are both implemented in a compiled language (C++), which is always faster than an interpreted language such as that used in the proposed package (Python).
However, all computation times remain quite reasonable, even for the methods that include a selection model step. 

\subsection{Robustness on misspecified models}
In this section, we address the question of the robustness of our estimator regarding the misspecification of the kernel function. More precisely, we generate a two-dimensional Hawkes process with power-law kernels, which are commonly used in the literature \citep{Mishra2016, Ogata1988} as an alternative to the exponential kernel for modelling a slower convergence to zero. For all $i, j$ we define the power-law kernel as: \[h_{ij}(t) = \frac{\alpha_{ij}\beta_{ij}}{(1 + \beta_{ij}t)^{1 + \gamma}}\,,\] with $\beta_{ij} > 0$, $\gamma > 0$ and $\alpha_{ij}\in\RR$ in order to allow inhibition effects. Let us remark that in the general case each kernel function could be given a different parameter $\gamma_{ij}$ but here we fix the same parameter for all interactions, which is a similar condition as Assumption~\ref{assu:beta}.

We propose to investigate different scenarios in order to model different behaviours. In all cases, we set the values of $\alpha_{ij}$ in order to have both excitation and inhibition effects and the values of $\mu_i$ as follows:
\[
 \begin{pmatrix}
  \mu_{1} \\
  \mu_{2} 
  \end{pmatrix} =
  \begin{pmatrix}
  1.0 \\ 
  1.0
  \end{pmatrix}\,,
  \quad
  \begin{pmatrix}
  \alpha_{11} & \alpha_{12}\\
  \alpha_{21} & \alpha_{22}
  \end{pmatrix}=
  \begin{pmatrix}
  0.1 & 1.5\\
  1.0 & -0.5
  \end{pmatrix}\,.
\]

\begin{itemize}
    \item  Scenarios $\gamma$ : we set
    
    \[\begin{pmatrix}
  \beta_{11} & \beta_{12}\\
  \beta_{21} & \beta_{22}
  \end{pmatrix}=
  \begin{pmatrix}
  1.0 & 1.1\\
  1.2 & 1.0
  \end{pmatrix}\ ,
\]
    
all values being similar in order to be close to Assumption~\ref{assu:beta}. Then we study the effect of parameter $\gamma$ which controls how the kernel functions decrease to zero. We set \[\gamma\in\{2.0, 4.0, 6.0, 8.0\}\,.\]

    \item Scenario $\beta$ : we keep the same values of $\mu_i$, $\alpha_{ij}$ as in Scenarios $\gamma$, we fix $\gamma=4.0$ and we choose very different values of $\beta_{i,j}$:

 \[\begin{pmatrix}
  \beta_{11} & \beta_{12}\\
  \beta_{21} & \beta_{22}
  \end{pmatrix}=
  \begin{pmatrix}
  1.0 & 2.0\\
  0.1 & 1.0
  \end{pmatrix}\,,
\]

    so that Assumption~\ref{assu:beta} is violated and therefore the intensity function of process $N^2$ is frequently non-monotonous between two event times.

\end{itemize}

We expect that our estimator should adapt better to Scenarios $\gamma$ (in particular large values of $\gamma$ would correspond to a fast convergence to zero) than to Scenario $\beta$.

Figure~\ref{fig:heatmap_powerlaw} represents the heatmap for $\sign(\alpha_{ij})\| h_{ij} \|_1 = \alpha_{ij}/\gamma$ and $\sign(\tilde \alpha_{ij})\| \tilde h_{ij}\|_1 = \tilde \alpha_{ij}/ \tilde\beta_{i}$ as well as whether the type of each interaction is correctly estimated (excitation or inhibition). We first notice that in most cases, our method is robust enough to differentiate between exciting and inhibiting interactions, the only errors concerning parameter $\alpha_{11}$ that is close to zero. For Scenarios $\gamma$, we can observe as expected that the bigger differences are obtained for smaller values of $\gamma$. Although the signs of the interactions are correctly estimated in Scenario $\beta$, we can observe substantial errors regarding the estimations of both $\alpha_{21}$ and $\alpha_{22}$. 

{\begin{figure*}[!ht]
     \centering
     \includegraphics[width=0.9\linewidth]{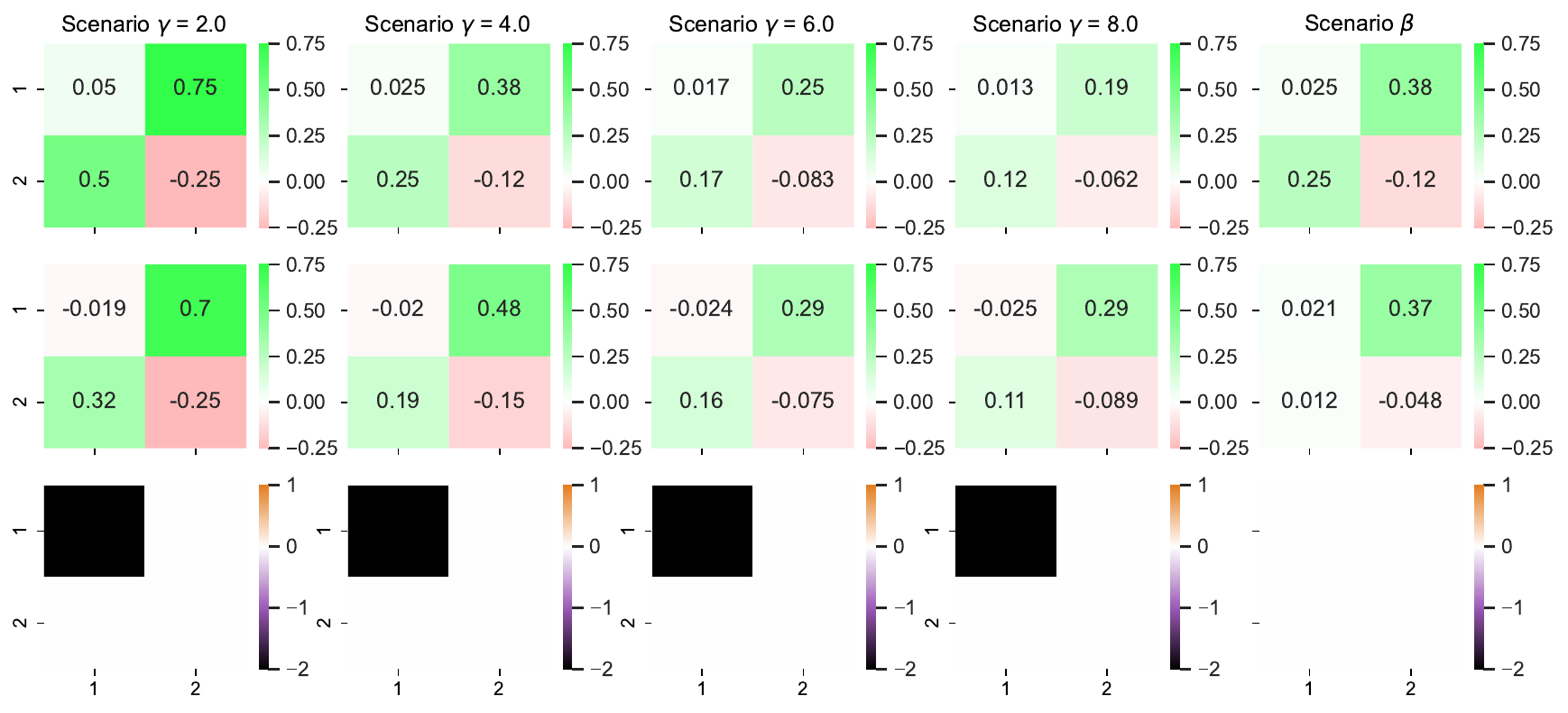}
     \caption{Top row corresponds to the heatmap $\sign(\alpha_{ij})\| h_{ij} \|_1 = \alpha_{ij}/\gamma$. Middle row corresponds to the estimated heatmap $\sign(\tilde \alpha_{ij})\| \tilde h_{ij}\|_1 = \tilde \alpha_{ij}/ \tilde\beta_{i}$. Bottom row corresponds to errors when estimating the nature of the interaction. A black box represents a value $\alpha_{ij}$ whose sign is wrongly estimated.}
     \label{fig:heatmap_powerlaw}
     \end{figure*}}

Table~\ref{tab:p_values_powerlaw} shows the average $p$-values associated with the goodness-of-fit measure presented in Section~\ref{sec:goodness}. It confirms that for all Scenarios $\gamma$, the $p$-values are smaller for smaller values of $\gamma$, all $p$-values remaining greater than 5\%. However, the $p$-values associated with Scenario $\beta$ are always equal to zero, which means that the goodness-of-fit is able to detect that the interactions are not correctly estimated. 

\begin{table}[!ht]
    \begin{center}
    \centering
    \begin{tabular}{c|c|ccc}
         \multicolumn{2}{c}{} & $p_1$ & $p_2$ & $p_{tot}$\\
         \toprule
         \multirow{5}{*}{Scenario $\gamma$} & $\gamma = 2.0$ & 0.140 & 0.275 & 0.068\\
         & $\gamma = 4.0$ & 0.484 & 0.536 & 0.482\\
         & $\gamma = 6.0$ & 0.381 & 0.503 & 0.408\\
         & $\gamma = 8.0$ & 0.408 & 0.506 & 0.330\\
         \midrule
         \multicolumn{2}{c|}{Scenario $\beta$} & \textbf{0.0} & \textbf{0.0} & \textbf{0.0}
    \end{tabular}
    \caption{Average $p$-values for estimations of two-dimensional Hawkes processes for all scenarios in the misspecified power-law model. The values are averaged over 25 simulations. In bold the $p$-values correspond to a rejected hypothesis at a confidence level of $0.95$.}
    \label{tab:p_values_powerlaw}
    \end{center}
    \end{table}

To conclude, if the true model is not too far from an exponential model and Assumption~\ref{assu:beta} holds, our procedure can adapt and provide reasonable estimations. If not, our estimator cannot adjust but we are able to detect incorrect estimations thanks to the goodness-of-fit procedure.

\section{Application on neuronal data}\label{sec:neuron}

\subsection{Preprocessing and data description}

In this section we present the results obtained by our estimation method applied to a collection of 10 trials consisting in the measurement of spike trains of 223 neurons from the lumbar spinal of a red-eared turtle. This data are first presented in  \cite{Peterson2016} and then also analysed in \cite{neuro2022} to study how the activity of a group of neurons impacts the membrane potential's dynamic of another neuron. In particular, recovering the connectivity graph allows to isolate a subnetwork which activity impacts the dynamics of one given neuron.  Events were registered for 40 seconds and in order to take into account eventual stationarity we only consider the events that took place on the interval $[11, 24]$ (see \cite{neuro2022} for further details).  Among all trials, each neuron recording contains between 54 and 4621 event times. Furthermore, we divide our samples in a training set consisting on all events in half the interval $[11, 17.5]$ and a test set consisting on the remaining window $[17.5, 24]$, in particular each neuron has at least 15 event times in each set. The training sets are used for obtaining the estimations and the test sets for performing the goodness-of-fit tests. 

\subsection{Resampling}

As only ten realisations are available, this can obviously limit the performance of both confidence intervals methods. In order to counter this problem, we perform a resampling method obtained as follows:
\begin{enumerate}
	\item We sample 3 realisations at random $(N_1, N_2, N_3)$, without replacement and by taking the order into account. From now on, we consider that each realisation takes place in the time interval $[0,6.5]$ (instead of $[11, 17.5]$).
	\item We cumulate all 3 realisations by considering that process $N_1$ takes place in $[0, 6.5]$ then process $N_2$ in $[6.5, 13]$ and finally $N_3$ in $[13, 19.5]$. This creates a single realisation $N$ in the interval $[0, 19.5]$. 
\end{enumerate}

This approach is proposed in \cite[Section 3.4]{Reynaud2014}. In our case, we repeat this process 20 times to obtain another sample of realisations. These new sample will be used for both the (MLE) method (presented as (resampled-MLE)) and for both (CfE) and (CfSt).

\subsection{Goodness-of-fit and multiple testing procedure}

Since we are dealing with a high-dimensional setting, it is crucial to account for a multiple testing correction which is performed through the Benjamini-Hochberg procedure as described in Section \ref{sec:support}. In this case, the adapted rejection threshold corresponds to $\frac{0.05 k}{d+1}$ and represented in Figure \ref{fig:p_values_neurons} by a blue line.
This is particularly useful in order to determine the best value of $\varepsilon$ for (MLE-$\varepsilon$) as we do not have prior knowledge regarding the sparsity of the neuronal connections.

{\begin{figure*}[!ht]
\centering
\includegraphics[width=0.9\linewidth]{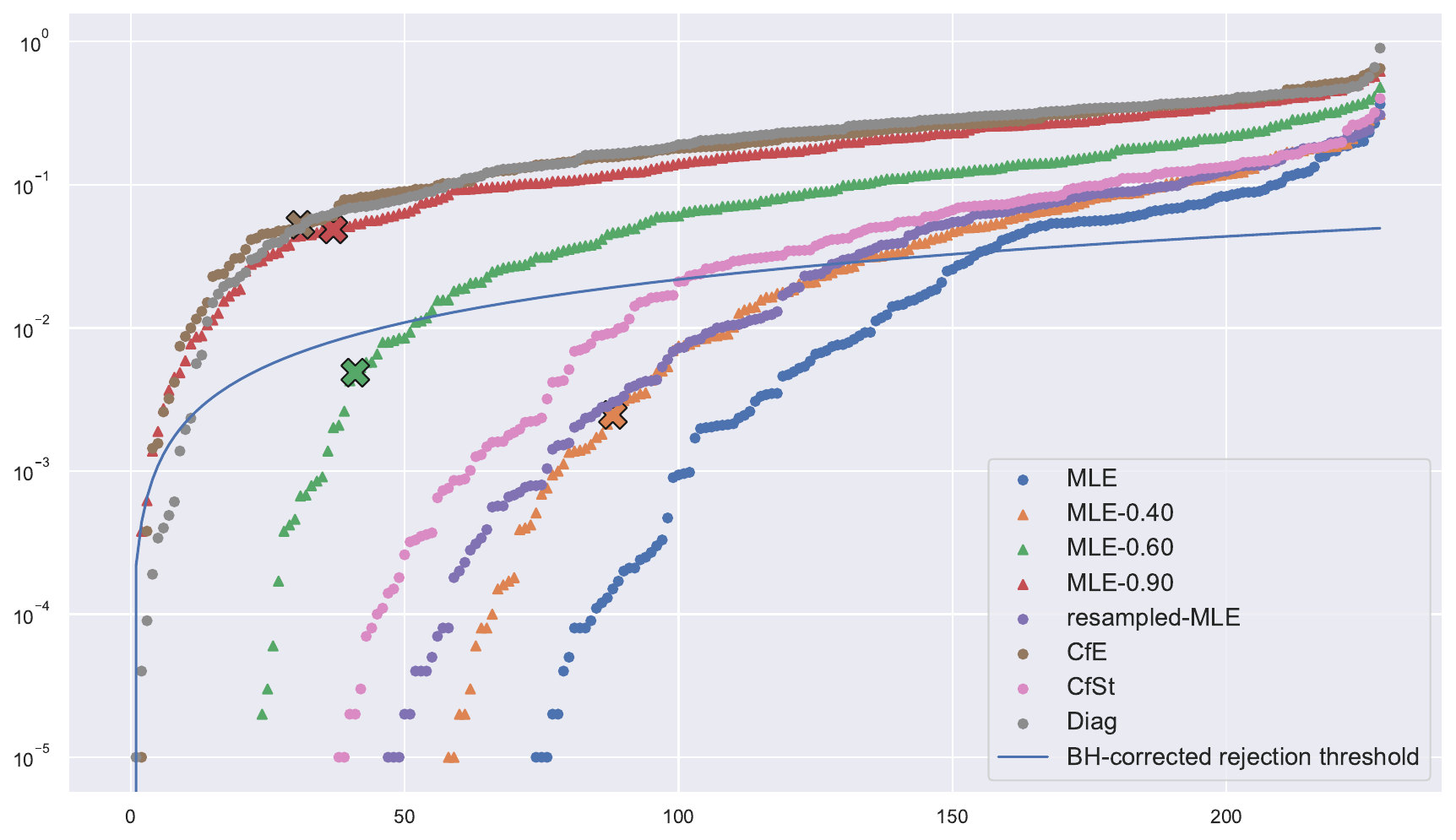}
\caption{Ordered $p$-values for all hypothesis tests $\mathcal{H}_i$ and $\mathcal{H}_{tot}$. $p_{tot}$ appears as a cross for each model (if there appears no cross for one given method, it means that the corresponding $p_{tot}$ is equal to zero). The blue curve corresponds to the adapted rejection threshold from the B-H procedure, so all tests whose $p$-value are under the line are rejected.}
\label{fig:p_values_neurons}
\end{figure*}}

 Figure~\ref{fig:p_values_neurons} shows the ordered $p$-values for each hypothesis $\mathcal{H}_i$ along with hypothesis $\mathcal{H}_{tot}$ displayed with a bold cross, all total $p$-values also being summarised in Table~\ref{tab:p_tot_neurons}. 
    
    \begin{table*}[!ht] 
    \begin{center}
    \setlength{\tabcolsep}{2pt}
    \centering
    \begin{tabular}{c|c|ccc|ccc|c}
          & MLE & MLE-$0.40$ & MLE-$0.60$ & MLE-$0.90$ & resampled-MLE & CfE & CfSt & Diag\\
         \toprule
         $p_{tot}$ & 0.0 & 0.002 & 0.004 & \textbf{0.048} & 0.0 & \textbf{0.053} & 0.0 & 0.0
    \end{tabular}
    \caption{Values of $p_{tot}$ for each estimation method for the neuronal dataset. In bold appear the \(p\)-values above the rejection threshold after Benjamini-Hochberg procedure.}
    \label{tab:p_tot_neurons}
    \end{center}
    \end{table*}
    
 We first note that for most methods, the $p$-values $p_{tot}$ associated with the  $\mathcal{H}_{tot}$ hypothesis are either equal to 0 or under the rejection threshold. In particular, it is the case for the MLE-$\varepsilon$ approach for small values of $\varepsilon$ (i.e. weak sparsity scenarios) but as we increase the threshold $\varepsilon$, the $p$-values appear to increase, with the best estimation being achieved for $\varepsilon = 0.90$.
 
 This suggests that the simpler the model the better $p$-values we obtain so we decided to include another approach, named ``Diag'', consisting in setting all $\alpha_{ij} = 0$ for $i \neq j$. This corresponds to a model where there exists no interaction between neurons and we keep only self-interactions: in other words, each neuron is seen as a univariate Hawkes process with three parameters $(\mu_i,\alpha_{ii}, \beta_i)$. Although most hypotheses $\mathcal{H}_i$ are not rejected by the method, the total $p$-value $p_{tot}$ is zero which suggests that although such a model could explain each dimension individually, it is unable to explain the neurons' interactions as a whole interconnected process.
Although the total $p$-values are generally quite low, they remain above the rejection threshold for (MLE-$0.90$) and (CfE).
Let us notice that the high value of the best threshold (0.90) is consistent with the (CfE) estimator that provides an even sparser solution with $4.26\%$ non-null interactions. 

Let us also recall that the (CfSt) method relies on the assumption that the MLE estimator is asymptotically normal. Here we have a low number of repetitions ($10$ trials) in a high-dimensional setting where some neurons are rarely observed which could explain why this estimator does not perform well. Moreover, in this context, a Kolmogorov-Smirnov test of normality is likely to provide high $p$-values for such small-sized samples even for non-normal distributions.
 
 Finally, the model that best describes the complete process $N$ is (CfE) with the highest value for $p_{tot}$ and with almost all hypotheses, including $\mathcal{H}_{tot}$, not rejected. This suggests indeed that the estimations provided by (CfE) are the best fit for explaining the entire process as well as each individual subprocess.

\subsection{Estimation results} \label{sec:comment_neuron}
Figure~\ref{fig:heatmap_CfE} illustrates the obtained estimation for all parameters for the (CfE) method. 
Let us recall that the estimation for (CfE) is obtained by using the resampled trials: the support is determined by using the empirical quantiles confidence intervals after an estimation through (MLE) and then all parameters are re-estimated over each trial. A single estimation is obtained by averaging over all trials.

Although the heatmap matrix corresponding to $(\sign(\tilde\alpha_{ij}))_{ij}$ contains only $4.26\%$ of non-null entries, there remain many significant interactions. Interestingly, among them we detect all types of interactions: mutual excitation, mutual inhibition, self-excitation, self-inhibition. This supports the relevance of carefully accounting for inhibition when developing inference procedures.

We also notice that the diagonal contains mostly non-null entries (all but $6$), which highlights the major effect of self-interactions, among which some are negative and some are positive. Although it is possible that different neurons actually show different patterns, some being self-exciting and other self-inhibiting, there exists another hypothesis.  We might indeed observe a combination of effects from which we cannot distinguish: on the one hand, a self-exciting behaviour and on the other hand, a refractory period following a spike during which a neuron cannot spike again. This could also explain why the order of magnitude of the $\beta_i$ estimations, which describe the duration until an effect vanishes, is different from a neuron to another. It would be of great interest to propose another modelling that could account for both effects and thereby helping us to provide additional information to support or refute this hypothesis.

Another striking phenomenon is the behaviour of neuron $13$, which seems to interact with many other neurons: it contains indeed $69\%$ non-null receiving interactions (row) and $57\%$ non-null giving interactions (column). Further analysis shows that this neuron spikes only in one out of ten trials so that it could indicate an inaccurate estimation. However, the $p$-value associated with this neuron's subprocess is not rejected by our goodness-of-fit procedure, which suggests that the corresponding estimation is actually accurate. Therefore, this neuron could either play central role among the whole network or be connected to an unobserved neuron with a central role. On the opposite side, some neurons exhibit only a few connections, in particular there is one neuron that only receives interactions without giving, while another one gives without receiving.

{\begin{figure*}[!ht]
\centering
\includegraphics[width=0.8\linewidth]{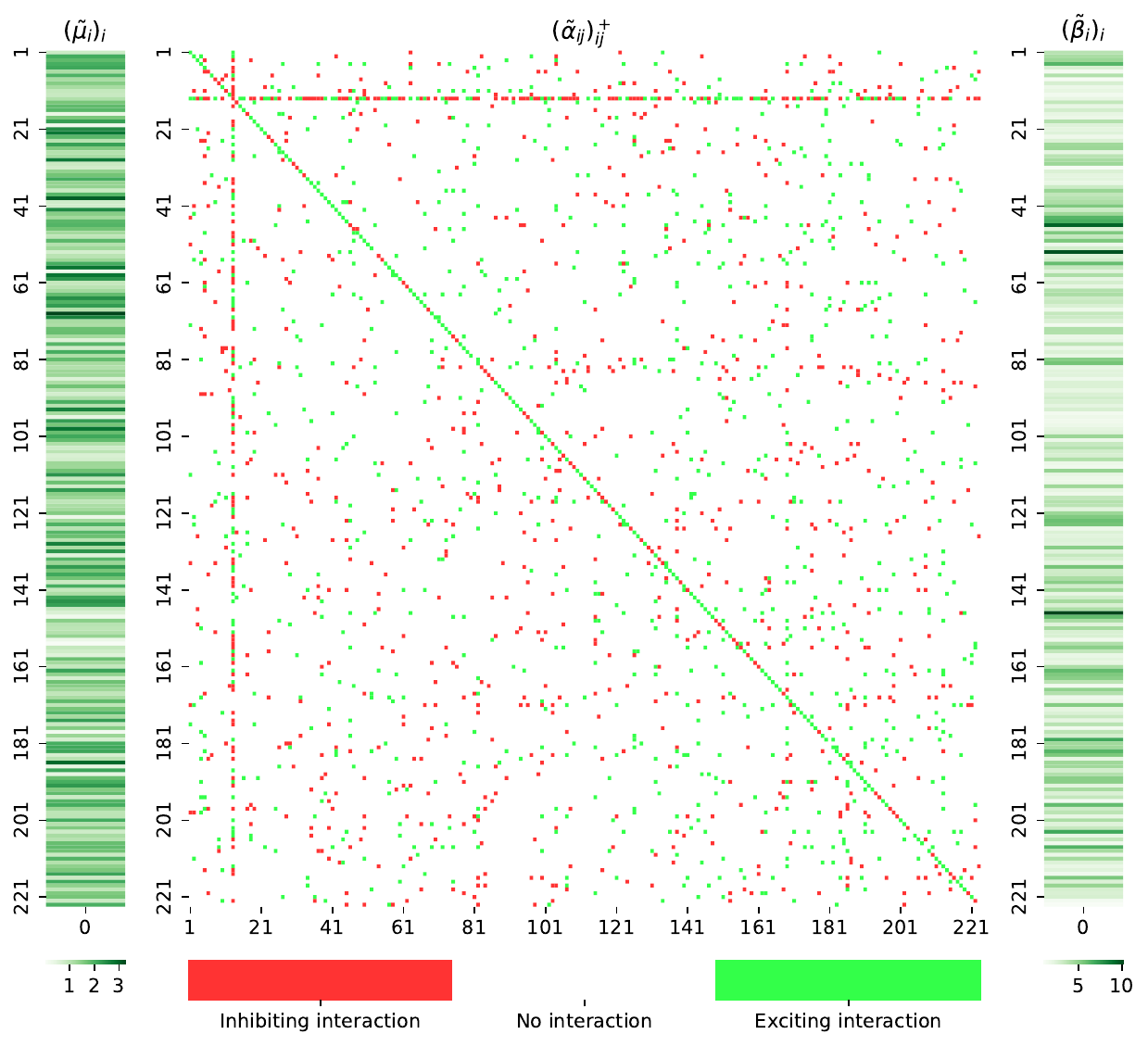}
\caption{Heatmap $(\sign(\alpha_{ij}))_{ij}$ of (CfE) estimation on 223 neurons.}
\label{fig:heatmap_CfE}
\end{figure*}}

\section{Discussion}
\label{sec:discussion}
In this paper, we proposed a methodology for estimating the parameters of multivariate exponential Hawkes processes with both exciting and inhibiting effects. Our first contribution was to provide a few sufficient conditions to ensure the identifiability of a commonly used model. Then we developed and implemented a maximum likelihood estimator combined with a variable selection procedure that enables to detect the significant interactions inside the whole process.
While our framework is more general than the usual linear Hawkes model, there remain two main limitations, the first one being the exponential distribution of the kernels, the second one the assumption that the delay factors \(\beta_{ij}\)
  only depend on the receiving process \(N^i\). If it is essential to assume a parametric form for the kernel functions in order to maintain the concepts of our approach, it would be of great interest to consider some extensions to account for potential combined effects, as already mentioned in Section \ref{sec:comment_neuron}. It would be notably relevant to include multi-scale effects or to consider a potential lag between an event time and its actual impact.  Regarding the assumption on the delay factors, while it is quite standard, it could be a limitation of our approach when considering heterogeneous phenomena.

  Going over this assumption would lead us to explore numerical integration methods and would considerably increase the computational time of the estimation procedure.
  This is obviously detrimental since, in practice, time sequences are increasingly abundant and large.
  On the other hand, improving the computational effectiveness of estimation procedures for Hawkes processes is a current direction of research \citep{Bompaire2018bis}.

  Our work focuses on the computational aspects of both maximum likelihood estimation and variable selection. It is of natural interest to provide further theoretical study of the asymptotic behaviour of our estimator, as done for exciting Hawkes processes \citep{Guo2018}.
  This work is currently under investigation.

  Let us also highlight that, because of the physical constraints of the experiment, only a fraction of the neuronal network is observed, which raises the question of interpretability of the estimated interactions.
  Indeed, the latter do not take into account the interactions with neurons that are outside the observed network.
  Very recent results tackle the consistency of estimated interactions in a partially observed network \citep{Reynaud2022}.
  A necessary condition to recover interactions in the subnetwork requires in particular to have a large number of interactions within the full network.
  Regarding the neuronal application, it could be of great interest to further investigate the interpretability of the inferred interactions and connectivity graph in light of the aforementioned work.

\bibliography{bibliography}

\pagebreak

\appendix

\section{Proof of Lemma~\ref{lemma:restart_times}} \label{appendix:proof_lemma}
  In order to prove Lemma~\ref{lemma:restart_times}, let us first state a preliminary result.
  
  \begin{lemma}\label{lemma:intensity_interval_form}
      If Assumption~\ref{assu:beta} is granted, then for each $i\in\{1,\ldots, d\}$ and any $k\geq 1$:
      \begin{gather*}
          \forall t \in [T\park, T\park[k+1]),
          \\
          \lambda^{i\star}(t) = \mu_i + \left(\lambda^{i\star}(T\park\right) - \mu_i)\mathrm{e}^{-\beta_i(t-T\park)}\,.
      \end{gather*}
    \end{lemma}

    \begin{proof}
      Let $i\in\{1,\ldots, d\}$. For any $k\geq 1$, the underlying intensity function $\lambda^{i\star}$ in the interval $[T\park, T\park[k+1])$ can be written:
      \[\lambda^{i\star}(t) = \mu_i + \sum_{j=1}^{d}{\sum_{\ell=1}^{N^j(t)}{\alpha_{ij}\mathrm{e}^{-\beta_{ij}(t-T_\ell^j)}}}\,.\]
      This function is differentiable in the open interval $(T\park, T\park[k+1])$ and we obtain:
      \[(\lambda^{i\star})'(t) = -\sum_{j=1}^{d}{\beta_{ij}\sum_{\ell=1}^{N^j(t)}{\alpha_{ij}\mathrm{e}^{-\beta_{ij}(t-T_\ell^j)}}}\,.\]

      By using Assumption~\ref{assu:beta} that for all $j\in\{1,\ldots, d\}$, $\beta_{ij} = \beta_i\in\RR_+^*$, we obtain the following differential equation:
      \[(\lambda^{i\star})'(t) = -\beta_i\left(\lambda^{i\star}(t) - \mu_i\right)\,,\] which by solving on the interval gives:
      \begin{equation*}
          \lambda^{i\star}(t) = \mu_i + \left(\lambda^{i\star}(T\park) - \mu_i\right)\mathrm{e}^{-\beta_i(t-T\park)}\,.
      \end{equation*}
    \end{proof}
    
    \begin{proof}[Proof of Lemma~\ref{lemma:restart_times}]
      Let $i\in\{1,\ldots, d\}$ and $k\geq 1$.
      By Lemma~\ref{lemma:intensity_interval_form}, $\forall t \in [T\park, T\park[k+1])$:
      \begin{equation*}
          \lambda^{i\star}(t) = \mu_i + \left(\lambda^{i\star}(T\park) - \mu_i\right)\mathrm{e}^{-\beta_i(t-T\park)}\,.
      \end{equation*}

      In particular, the derivative of the underlying intensity function is of opposite sign as $(\lambda^{i\star}(T\park) - \mu_i)$.
      Let us distinguish two cases,
      referring to Figure~\ref{fig:lemma} for a better understanding:
      
      {\begin{figure*}[!ht]
      \centering
      \includegraphics[width=\linewidth]{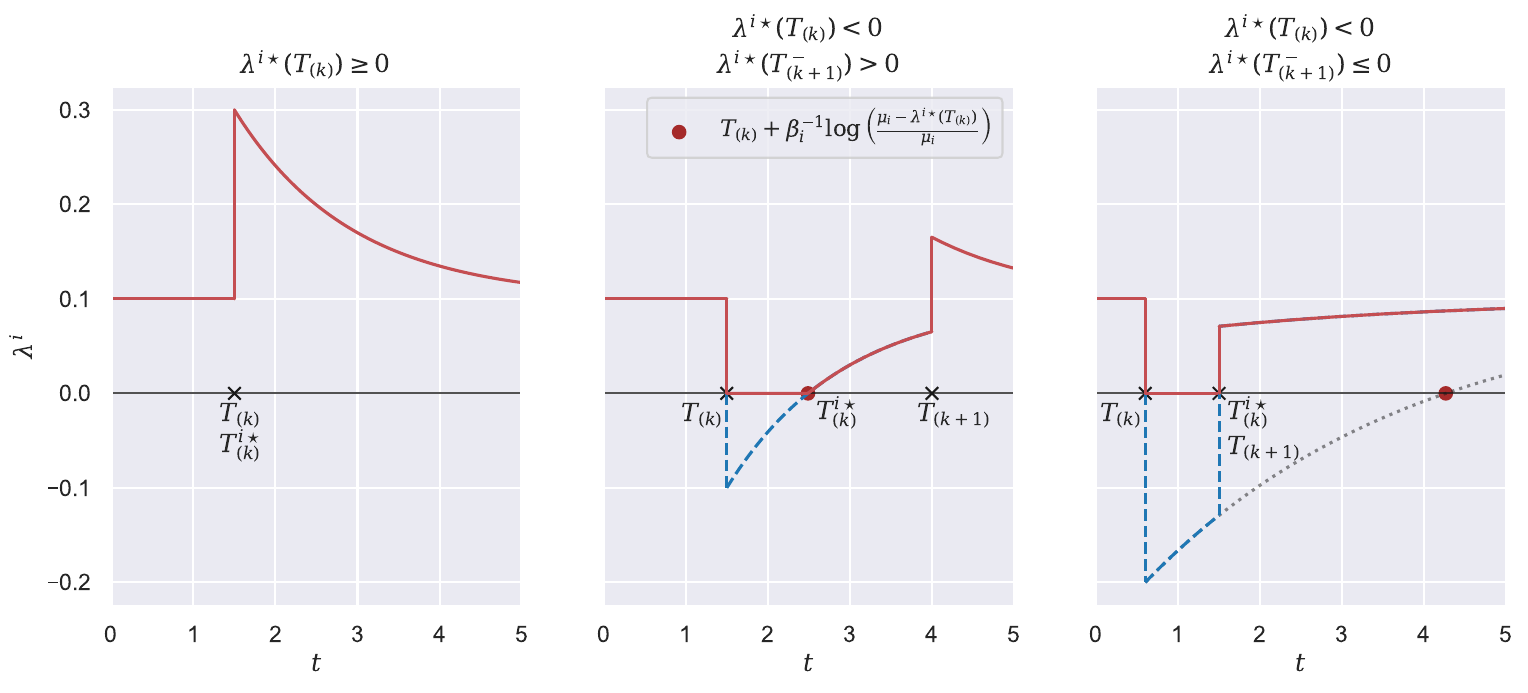}
      \caption{Illustration of three possible scenarios for restart times $T\park^{i\star}$ depending on the sign of $\lambda^{i\star}(T\park)$ and $\lambda^{i\star}(T\park[k+1]^-) = \lim_{t\to T\park[k+1]^-}{\lambda^{i\star}(t)}$. The dotted line in the last scenario shows the equation $\mu_i + (\lambda^{i\star}(T\park) - \mu_i)\mathrm{e}^{-\beta_i(t-T\park)} = 0$ and the term $T\park + \beta_i^{-1}\log{\adaptedpar{\frac{\mu_i - \lambda^{i\star}(T\park)}{\mu_i}}}$ as its only root.}
      \label{fig:lemma}
      \end{figure*}}

      \begin{itemize}
          \item If $\lambda^{i\star}(T\park) \geq 0$, then,
          \begin{align*}
            T^{i\star}\park = T\park &= \min{\adaptedpar{T\park, T\park[k+1]}} \\
            &= \min{\adaptedpar{t^\star_k, T\park[k+1]}}\,.
          \end{align*}

          If $(\lambda^{i\star}(T\park) - \mu_i) \geq 0$, then $\lambda^{i\star}$ is decreasing and lower-bounded by $\mu_i$. If $(\lambda^{i\star}(T\park) - \mu_i) < 0$ then $\lambda^{i\star}$ is increasing and lower-bounded by zero. In both cases, for any $t\in(T^{i\star}\park, T\park[k+1])$, $\lambda^{i\star}(t) > 0$ and then $\lambda^i(t) = \lambda^{i\star}(t)$.

          \item If $\lambda^{i\star}(T\park) < 0$, then $\left(\lambda^{i\star}\left(T\park\right) - \mu_i\right) < 0$ so $\lambda^{i\star}$ is strictly increasing and by continuity and by Lemma~\ref{lemma:intensity_interval_form}, there exists a unique $t^\star > T\park$ such that $\mu_i + \left(\lambda^{i\star}(T\park\right) - \mu_i)\mathrm{e}^{-\beta_i(t^\star-T\park)} = 0$.

              We obtain:
              \[t^\star = T\park + \beta_i^{-1}\log{\adaptedpar{\frac{\mu_i - \lambda^{i\star}(T\park)}{\mu_i}}}\,.\] By denoting $\lambda^{i\star}(T\park[k+1]^-) := \lim_{t\to T\park[k+1]^-}{\lambda^{i\star}(t)}$:

          \begin{itemize}
              \item If $\lambda^{i\star}(T\park[k+1]^-) > 0$, then $t^\star < T\park[k+1]$ by strict increasingness and so by definition $T^{i\star}\park = t^\star = t^\star_k$.
              Lastly, for any $t\in(T\park, T\park[k+1])$,
              if $t\in(T\park, T^{i\star}\park]$, $\lambda^{i\star}(t) \leq 0$ and then $\lambda^i(t) = 0$,
              while if $t\in(T^{i\star}\park, T\park[k+1])$, $\lambda^{i\star}(t) > 0$ and then $\lambda^i(t) = \lambda^{i\star}(t)$.
              
              \item If $\lambda^{i\star}(T\park[k+1]^-) \leq 0$, then by strict increasingness $t^\star > T\park[k+1]$ and so $T\park^{i\star} = T\park[k+1]$.
              In this case, for all $t\in(T\park, T\park[k+1])$,
              $\lambda^{i\star}(t) < 0$ so $\lambda^i(t)=0$.
              Moreover, $(T^{i\star}\park, T\park[k+1]) = \emptyset$ so we never have $\lambda^i(t) = \lambda^{i\star}(t)$.
          \end{itemize}

      \end{itemize}

      Combining all scenarios achieves the proof.
    \end{proof}
    
\section{Proof of Proposition~\ref{prop:compensator}} \label{appendix:proof_prop_compensator}
    
    \begin{proof}
      For each $i\in\{1,\ldots, d\}$ and \(\forall t\geq 0\), with convention \(T\park[0] = 0\) and \(T\park[0]^{i\star} = 0\):
      \begin{align*}
        \Lambda^i(t)
        &= \int_{0}^{t}{\lambda^i(u)\,\mathrm{d}u}
        = \sum_{k=0}^{N(t)} \int_{T\park}^{T\park[k+1]}{\lambda^i(u) \II_{u \le t} \,\mathrm{d}u}\\
        &= \sum_{k=0}^{N(t)} \int_{T\park^{i\star}}^{T\park[k+1]} \lambda^{i\star}(u) \II_{u \le t} \, \mathrm{d}u \,,
      \end{align*}
      where the last equation comes from Lemma~\ref{lemma:restart_times}.
      Then, for \(k=0\):
      \begin{align*}
        \int_{T\park[0]^{i\star}}^{T\park[1]} \lambda^{i\star}(u) \II_{u \le t} \, \mathrm{d}u
        &= \int_{T\park[0]^{i\star}}^{T\park[1]} \mu^i \II_{u \le t} \, \mathrm{d}u\\
        &= \mu_i \min(t, T\park[1]) \,,
      \end{align*}
      and for every \(k \in \{1, \dots, N(t)\}\), by Lemma~\ref{lemma:intensity_interval_form}:
      \begin{align*}
        &\int_{T\park^{i\star}}^{T\park[k+1]} \lambda^{i\star}(u) \II_{u \le t} \, \mathrm{d}u \\
        = &\int_{T\park^{i\star}}^{T\park[k+1]} \left[ \mu_i + \left(\lambda^{i\star}(T\park\right) - \mu_i)\mathrm{e}^{-\beta_i(t-T\park)} \right] \II_{u \le t} \, \mathrm{d}u \\
        = &\mu_i \left(\min(t, T\park[k+1]) - T\park^{i\star}\right)\\ &+ \beta_i^{-1} \left(\lambda^{i\star}(T\park) - \mu_i \right) \\
        &\left(\mathrm{e}^{-\beta_i(T\park^{i\star} - T\park)} - \mathrm{e}^{-\beta_i(\min(t, T\park[k+1])- T\park)} \right)\,.
      \end{align*}
    \end{proof}

\section{Proof of Theorem~\ref{theorem:identifiability}}\label{appendix:proof}

\begin{proof}
    We only need to prove that if $\lambda_{\theta_i}^i(t \mid \mathcal{F}_t) = \lambda_{\theta'_i}^i(t \mid \mathcal{F}_t)$ a.e. for every $i\in\{1, \ldots, d\}$ then $\theta = \theta'$. In this proof, both intensities are considered with respect to the same filtration $\mathcal{F}_t$ so we will omit it from the rest of the proof.
    
    Let $\theta, \theta'\in \Theta$. Let us assume that $\lambda_{\theta_i}^i(t) = \lambda_{\theta'_i}^i(t)$ for all $t < T$. In order to prove equality between the two parameters we will first prove that $\mu_i = {\mu_i}'$, then $\beta_i = \beta_i'$ and lastly that $\alpha_{ij} = \alpha_{ij}'$ for every $i,j$. 
    
    \begin{itemize}
        \item For any $i$, as $\lambda_{\theta_i}^i(t) = \lambda_{\theta'_i}^i(t)$ a.e., then 
        \begin{alignat*}{2}
            &\quad& \Lambda_{\theta_i}^i(T\park[1]) &= \Lambda_{\theta'_i}(T\park[1])\\
            \implies && \mu_i T\park[1] &= {\mu_i}' T\park[1]\\
            \implies && \mu_i &= {\mu_i}'\,.
        \end{alignat*}
        
        \item For any $i$, let us choose $T\park[k+1]$ such that it is an event of process $N^i$. 
        As \[\PP(T\park[k+1] \text{ is an event of }N^i \mid \mathcal{F}_{T\park[k+1]^-}) = \frac{\lambda^i_{\theta_i}(T\park[k+1]^-)}{\lambda_{\theta}(T\park[k+1]^-)}\,,\]
        then $\lambda_{\theta_i}^i(T\park[k+1]^-) > 0$.
        By definition of $T_{(k),\theta}^{i\star}$, \[T_{(k),\theta}^{i\star} < T\park[k+1] \text{ a.s.} \,.\]
        Furthermore, $\lambda_{\theta_i}^i(t) = 0$ for $t\in(T\park, T_{(k), \theta_i}^{i\star})$ and $\lambda_{\theta_i}^i(t) > 0$ for $t\in(T_{(k), \theta_i}^{i\star}, T\park[k+1])$.
        Then, as we assumed that $\lambda_{\theta_i}^i(t) = \lambda_{\theta'_i}^i(t)$ a.e., we can conclude that $T_{(k), \theta_i}^{i\star} = T_{(k), \theta'_i}^{i\star}$. 
        By differentiating the intensity functions on the interval $(T_{(k),\theta_i}^{i\star}, T\park[k+1])$ as in the proof of Lemma \ref{lemma:restart_times} we obtain: 
        
        \begin{align*}
            \quad&(\lambda_{\theta_i}^i)'(t) = (\lambda_{\theta'_i}^i)'(t) \text{ a.e.}\\ 
            \implies &-\beta_i(\lambda_{\theta_i}^i(t) - \mu_i) =  -\beta_i'(\lambda_{\theta'_i}^i(t) - \mu_i) \text{ a.e.} \\
            \implies &(\beta_i - \beta_i')(\lambda_{\theta_i}^i(t) - \mu_i) = 0  \text{ a.s.}\\\
            \implies &\int_{T_{(k), \theta}^{i\star}}^{T\park[k+1]}{(\beta_i - \beta_i')(\lambda_{\theta_i}^i(t) - \mu_i)\,\mathrm{d}t} = 0\\
            \implies &(\beta_i - \beta_i')\int_{T_{(k), \theta}^{i\star}}^{T\park[k+1]}{(\lambda_{\theta_i}^i(t) - \mu_i)\,\mathrm{d}t} = 0\,.
        \end{align*}
        
        Additionnally,  $\lvert \lambda_{\theta_i}^i(t) - \mu_i \rvert > 0$ a.s. for all $t\in(T_{(k),\theta_i}^{i\star}, T\park[k+1])$ as $\lambda^i_{\theta_i}$ is monotone and converges to $\mu_i$. Then it follows that the integral is non-zero and so $\beta_i = \beta_i'$.        
    
    \item Let us prove the equality $\alpha_{ij} = \alpha_{ij}'$. 
    By the assumption made on the event times, for any $i,j$ with $i\neq j$, there exists two event times $\tau < \tau_+$ with $\tau$ an event time from process $N_j$ and $\tau_+$ an event time from process $N_i$ such that:
    \begin{enumerate}
        \item $\lambda_{\theta_i}^i(\tau^-) > 0$;
        \item there are only events of process $N^j$ in the interval $[\tau, \tau_+)$.
    \end{enumerate} 
    
    Let $\tau$ and $\tau_+$ be two such event times. 
    As a reminder, $T\park[N(\tau) - 1]$ corresponds to the event time before $\tau$ and similarly for $\tau_+$. 
    For $t \in [T\park[N(\tau) - 1], \tau)$, as $\lambda_{\theta_i}^i(t) = \lambda_{\theta'_i}^i(t)$ a.e.\ and by Condition \ref{hyp:ii}:
    
    \begin{equation*}
    	\lambda_{\theta_i}^{i\star}(\tau^-) = \lambda_{\theta'_i}^{i\star}(\tau^-) \,.
    \end{equation*}

    By using the following equalities (proven in the previous points of the proof) 
    \begin{gather*}
    \mu_i = {\mu_i}'\,, \qquad \beta_i = \beta_i' \,,\\
    T_{(N(\tau) - 1), \theta_i}^{i\star} = T_{(N(\tau) - 1), \theta'_i}^{i\star}\,,
    \end{gather*} 
    it follows that 
    \begin{gather}
    \lambda_{\theta_i}^{i\star}(\tau^-) = \lambda_{\theta'_i}^{i\star}(\tau^-)\nonumber\\
    \implies \sum_{l=1}^{d}{\alpha_{il}\sum_{T_k^l < \tau}{\mathrm{e}^{-\beta_i(\tau - T_k^l)}}} =\sum_{l=1}^{d}{\alpha_{il}'\sum_{T_k^l < \tau}{\mathrm{e}^{-\beta_i(\tau - T_k^l)}}}\nonumber\\
    \implies \sum_{l=1}^{d}{(\alpha_{il} - \alpha_{il}')A^l}= 0\,,\label{eq:A^l}
    \end{gather}
    where \[A^l = \sum_{T_k^l < \tau}{\mathrm{e}^{-\beta_i(\tau - T_k^l)}}\,.\]
    
    We can then write Equation \eqref{eq:A^l} by replacing $\tau$ by $\tau_+$ as $\lambda_{\theta_i}^i(\tau_+^-) > 0$ because $\tau_+$ is an event time of process $N^i$.  We obtain then
    \begin{equation}\label{eq:B^l}
        \sum_{l=1}^{d}{(\alpha_{il} - \alpha_{il}')B^l} = 0\,,
    \end{equation}
    where \[B^l = \sum_{T_k^l < \tau_+}{\mathrm{e}^{-\beta_i(\tau_+ - T_k^l)}}\,.\]
    
   By definition of event times $\tau$ and $\tau_+$, all events on the interval $[\tau, \tau_+)$ are from process $N^j$ so we obtain for all $l \neq j$, \[B^l = A^l \mathrm{e}^{-\beta_i(\tau_+ - \tau)}\,.\] 
    For $l=j$, 
    \[B^j = A^j\mathrm{e}^{-\beta_i(\tau_+ - \tau)} + \sum_{\tau \leq T_k^j < \tau_+}{\mathrm{e}^{-\beta_i(\tau_+ - T_k^j)}}\,,\] 
    where the second term of the right hand side is positive as interval $[\tau, \tau_+)$ contains at least one event, $\tau$, from process $N^j$. 
    We can rewrite then Equation \eqref{eq:B^l}:
    \begin{equation*}
    \sum_{l=1}^{d}{(\alpha_{il} - \alpha_{il}')B^l} = 0
    \end{equation*}
    \begin{align*}
    \implies &\mathrm{e}^{-\beta_i(\tau_+ - \tau)}\sum_{l=1}^{d}{(\alpha_{il} - \alpha_{il}')A^l} \\
    &+ (\alpha_{ij}-\alpha_{ij}')\sum_{\tau \leq T_k^j < \tau_+}{\mathrm{e}^{-\beta_i(\tau_+ - T_k^j)}} = 0\,.
    \end{align*}
        
    By Equation \eqref{eq:A^l}, the first term is null and the second sum is non-zero as interval $[\tau, \tau_+)$ contains at least an event from process $N^j$. It follows that:
    
    \begin{gather*}
        (\alpha_{ij}-\alpha_{ij}')\sum_{\tau \leq T_k^j < \tau_+}{\mathrm{e}^{-\beta_i(\tau_+ - T_k^j)}} = 0\\
        \implies \quad (\alpha_{ij}-\alpha_{ij}') = 0\,.
    \end{gather*}
    
    It follows that for every $j\neq i$, $\alpha_{ij}=\alpha_{ij}'$. It remains to prove that $\alpha_{ii}=\alpha_{ii}'$. For this, let $T^i_2$ be the second event time of process $N^i$. We can the write the equality 
    
        \begin{equation*}
    	\lambda_{\theta_i}^{i\star}(T_2^i) - \lambda_{\theta'_i}^{i\star}(T^i_2) = 0 \,,
    \end{equation*}
    \begin{equation*}
    	 \implies \sum_{l=1}^{d}{(\alpha_{il} - \alpha'_{il}) \sum_{T_k^l < T^i_2}{\mathrm{e}^{-\beta_i(\tau - T_k^l)}}}= 0 \,,
    \end{equation*}
    \begin{equation*}
    	 \implies (\alpha_{ii} - \alpha'_{ii}) \sum_{T_k^i < T^i_2}{\mathrm{e}^{-\beta_i(\tau - T_k^l)}}= 0 \,,
    \end{equation*}
    as for $j\neq i$, $\alpha_{ij}=\alpha_{ij}'$. The sum is non-zero as it contains the event $T^i_1$ and so we obtain $\alpha_{ii}=\alpha_{ii}'$. 

\end{itemize}
This achieves the proof.
\end{proof}

\section{Proof of Corollary~\ref{cor:mle}} \label{appendix:proof_cor_mle}
    \begin{proof}
          For all \(i \in \{1, \dots, d\}\), \(\theta \in \Theta\) and \(k \in \mathbb N^*\),
          \begin{align*}
            \log{\lambda^i_{\theta_i}(T_k^{i-})}
            = &- \infty \II_{\lambda^i_{\theta_i}(T_k^{i-}) = 0} \\
            &+ \log{\lambda^i_{\theta_i}(T_k^{i-})} \II_{\lambda^i_{\theta_i}(T_k^{i-}) > 0} \\
            = &- \infty \II_{\lambda^{i\star}_{\theta_i}(T_k^{i-}) \le 0} \\
            &+ \log{\lambda^{i^\star}_{\theta_i}(T_k^{i-})} \II_{\lambda^{i\star}_{\theta_i}(T_k^{i-}) > 0} \\
            = &\log{\lambda^{i^\star}_{\theta_i}(T_k^{i-})} \\
            = &\log{\lim_{t \to T_k^{i-}} \lambda^{i^\star}_{\theta_i}(t)}.
          \end{align*}
          Now, for \(k=1\), \(\lambda^{i^\star}_{\theta_i}(T_k^{i-}) = \mu_i\),
          and for \(k \geq 2\), let us note \(q = N(T_k^i)-1\).
          Then,
          \[
            [T\park[q], T\park[q+1]) = [T\park[N(T_k^i)-1], T\park[N(T_k^i)]) = [S_k^i, T_k^i) \,,
          \]
          and by Lemma~\ref{lemma:intensity_interval_form}, \(\forall t \in [T\park[q], T\park[q+1])\):
          \begin{align*}
            \lambda_\theta^{i\star}(t)
            &= \mu_i + \left(\lambda_{\theta_i}^{i\star}(T\park[q]) - \mu_i\right)\mathrm{e}^{-\beta_i(t-T\park[q])}\\
            &= \mu_i + \left(\lambda_{\theta_i}^{i\star}(S_k^i) - \mu_i\right)\mathrm{e}^{-\beta_i(t-S_k^i)}\,.
          \end{align*}
          Thus, 
          \begin{align*}
            \lim_{t \to T_k^{i-}} \lambda^{i^\star}_{\theta_i}(t)
            &= \lim_{t \to T\park[q]^{-}} \lambda^{i^\star}_{\theta_i}(t)\\
            &= \mu_i + \left(\lambda_{\theta_i}^{i\star}(S_k^i) - \mu_i\right)\mathrm{e}^{-\beta_i(T_k^{i}-S_k^i)}\,.
          \end{align*}
          To conclude, by Equation~\eqref{eq:general_log_likelihood},
          \begin{align*}
            \ell_t^{i}(\theta_i)
            = &\sum_{k=1}^{N^i(t)}{\log{\lambda^i_{\theta_i}(T_k^{i-})}} - \Lambda_{\theta_i}^i(t) \\
            = &\log{\lambda^i_{\theta_i}(T_1^{i-})} + \sum_{k=2}^{N^i(t)}{\log{\lambda^i_{\theta_i}(T_k^{i-})}} - \Lambda_{\theta_i}^i(t) \\
            = &\log{\mu_i} + \sum_{k=2}^{N^i(t)}{\log \left(\mu_i + \left(\lambda_{\theta_i}^{i\star}(S_k^i) - \mu_i\right)\mathrm{e}^{-\beta_i(T_k^{i}-S_k^i)} \right)}\\
            &- \Lambda_{\theta_i}^i(t) \,.
          \end{align*}
    \end{proof}

\section{Algorithm for computing the log-likelihood}
\label{app:log-lik}

This section presents Algorithm \ref{alg:likelihood} for computing the log-likelihood $\ell_t(\theta)$ by leveraging the results from Corollary \ref{cor:mle}.
  \begin{algorithm}[!ht]
  \SetAlgoLined
   \textbf{Input} Parameters $\mu_i$, $\alpha_{ij}$, $\beta_i$ for $i,j\in\{1,\ldots, d\}$, list of event times and marks $(T\park, m_k)_{k=1:N(t)}$\;
   \textbf{\underline{Initialisation}} Initialise for all $i$, $\Lambda^i_k = \mu_i T\park[1]$, $\lambda^{i\star}(T\park^-) = \mu_i$, $\lambda_k^{i\star} = \mu_i + \alpha_{im_1}$ and $\ell_t(\theta) = \log(\lambda^{m_1\star}(T\park^-)) - \sum_{i=1}^{d}{\Lambda^i_k}$\;
   \For{$k = 2 \text{ to } N(t)$}{

   Compute for all $i$, $T\park[k-1]^{i\star} = \min{\adaptedpar{T\park[k-1] + \beta_i^{-1}\log{\adaptedpar{\frac{\mu_i - \lambda_k^{i\star}}{\mu_i}}}\II_{\{\lambda_k^{i\star} < 0\}}, T\park}}$\;

   Compute for all $i$, $\Lambda^i_k = \mu_i(T\park - T\park[k-1]^{i\star}) + \beta_i^{-1}(\lambda_k^{i\star} - \mu_i)(\mathrm{e}^{-\beta_i(T\park[k-1]^{i\star} - T\park[k-1])} - \mathrm{e}^{-\beta_i(T\park- T\park[k-1])})$\;

   Compute for all $i$, $\lambda^{i\star}(T\park^-) = \mu_i + (\lambda^{i\star}_k - \mu_i)\mathrm{e}^{-\beta_i(T\park - T\park[k-1])}$\;

   Update $\ell_t(\theta) = \ell_t(\theta) + \log(\lambda^{m_k\star}(T\park^-)) - \sum_{i=1}^{d}{\Lambda^i_k}$\;

   Compute for all $i$, $\lambda_k^{i\star} = \lambda^{i\star}(T\park^-) + \alpha_{im_k}$\;

   }
   Compute for all $i$, $T\park[N(t)]^{i\star} = \min{\adaptedpar{T\park[N(t)] + \beta_i^{-1}\log{\adaptedpar{\frac{\mu_i - \lambda_k^{i\star}}{\mu_i}}}\II_{\{\lambda_k^{i\star} < 0\}}, t}}$\;

   Compute for all $i$, $\Lambda^i_k = \left[\mu_i(t - T\park[N(t)]^{i\star}) + \beta_i^{-1}(\lambda_k^{i\star} - \mu_i)(\mathrm{e}^{-\beta_i(T_{N(t)}^{i\star} - T\park[N(t)])} - \mathrm{e}^{-\beta_i(t- T\park[N(t)])})\right]\II_{\{t > T\park[N(t)]^{i\star}\}}$\;

   Update $\ell_t(\theta) =  \ell_t(\theta) - \sum_{i=1}^{d}{\Lambda^i_k}$\;

   \Return Log-likelihood $\ell_t(\theta)$.
   \caption{Computation of the log-likelihood $\ell_t(\theta)$ of a multivariate exponential Hawkes process.}
   \label{alg:likelihood}
  \end{algorithm}

\section{Reconstructed interaction functions for synthetic data}
\label{app:numerical}

This section presents the reconstruction of interaction functions $h_{ij}$ along with the estimated functions $\tilde h_{ij}$ from the two-dimensional Hawkes processes simulations as described in Section \ref{sec:dim2}. Figure \ref{fig_supp:scenario_1} and Figure \ref{fig_supp:scenario_2} correspond respectively to the estimations for Scenarios (1) and (2) from Table \ref{tab:simulated_2_table}.

{\begin{figure}[!ht]
     \centering
     \includegraphics[width=\linewidth]{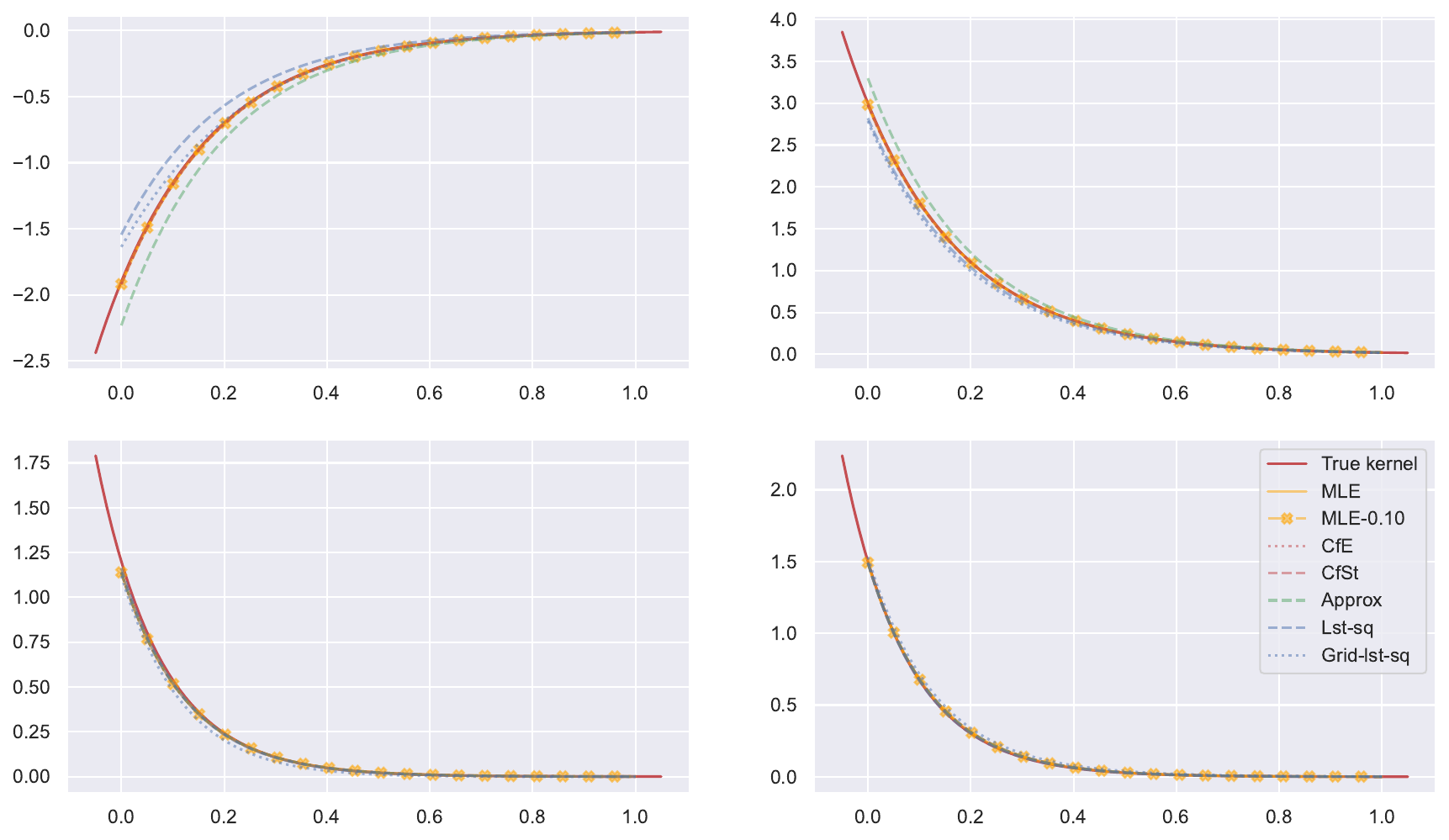}
     \caption{Reconstruction of interaction functions $h_{ij}$ for Scenario (1) of two-dimensional Hawkes processes along with all estimated functions $\tilde h_{ij}$. The real function is plotted in red and 25 estimations are averaged for each method.}
     \label{fig_supp:scenario_1}
     \end{figure}}
     
{\begin{figure}[!ht]
     \centering
     \includegraphics[width=\linewidth]{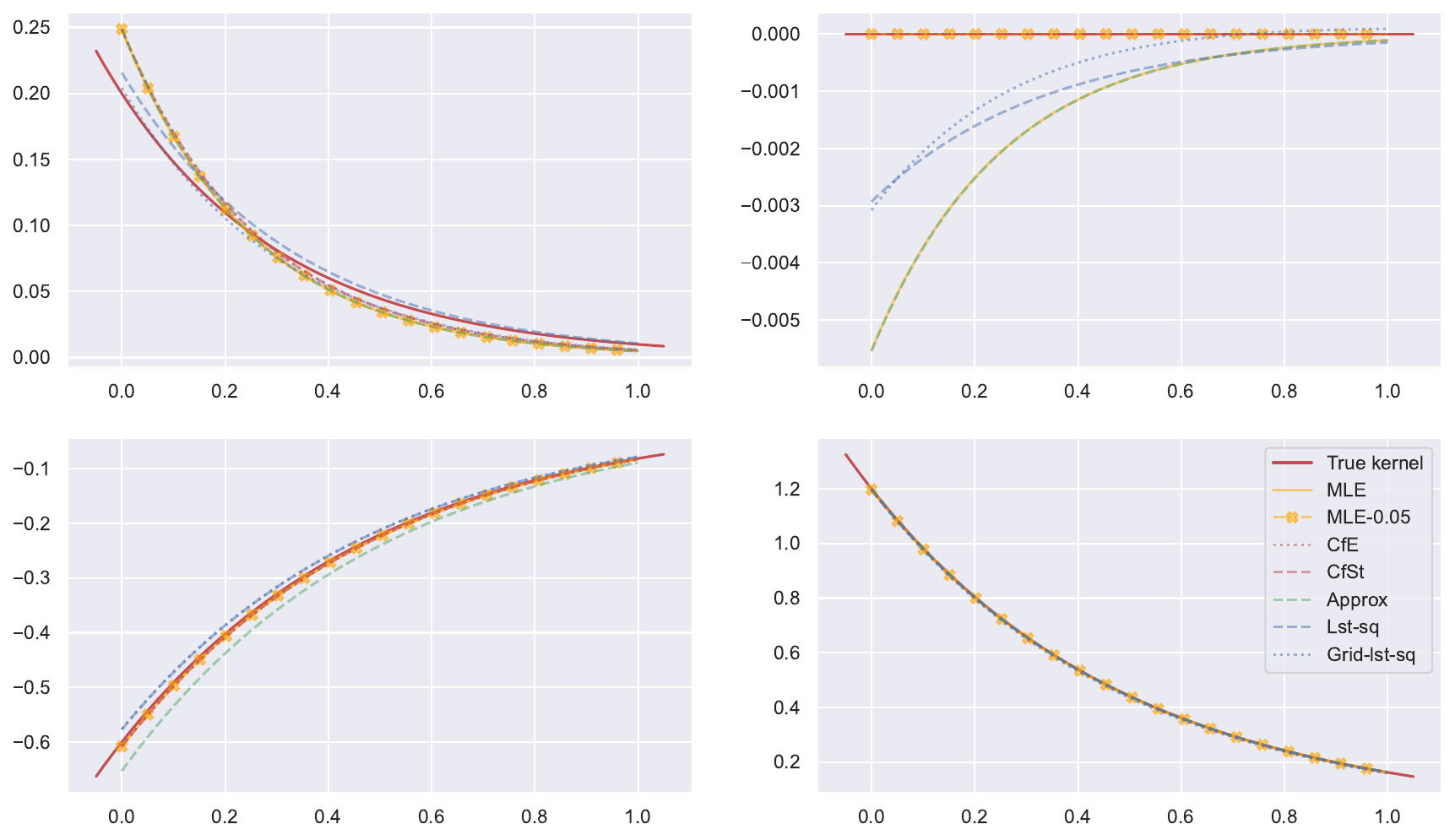}
     \caption{Reconstruction of interaction functions $h_{ij}$ for Scenario (2) of two-dimensional Hawkes processes along with all estimated functions $\tilde h_{ij}$. The real function is plotted in red and 25 estimations are averaged for each method.}
     \label{fig_supp:scenario_2}
     \end{figure}}

\end{document}

%% file: commands.tex
\usepackage[a4paper,top=3cm,bottom=3cm,left=3cm,right=3cm,marginparwidth=2cm]{geometry}

\usepackage[english]{babel}
\usepackage[utf8]{inputenc}
\usepackage[T1]{fontenc}

\usepackage{mathtools,amsmath,amsthm,amsfonts,amssymb,bm}
\usepackage{graphicx}
\usepackage[colorlinks=true, allcolors=blue]{hyperref}
\usepackage{dsfont}
\usepackage{mathrsfs}
\usepackage{csquotes}
\usepackage[ruled,vlined]{algorithm2e}
\usepackage{subfig}
\usepackage{caption}
\usepackage{multirow}
\usepackage{array, booktabs}
\usepackage{enumitem}
\usepackage{cleveref}

\usepackage{ulem}
\usepackage[dvipsnames]{xcolor}
\definecolor{Red}{RGB}{220, 20, 60}
\usepackage[colorinlistoftodos,textwidth=2.5cm]{todonotes}

\theoremstyle{definition}

\theoremstyle{plain}
\newtheorem{theorem}{Theorem}[section]
\newtheorem{proposition}{Proposition}[section]
\newtheorem{corollary}{Corollary}[proposition]
\newtheorem{lemma}{Lemma}[section]
\newtheorem{remark}{Remark}[section]
\newtheorem{assumption}{Assumption}
\newtheorem{example}{Example}

\def\NN{{\mathbb N}} %Letters and numbers
\def\PP{{\mathbb P}}
\def\EE{{\mathbb E}}
\def\RR{{\mathbb R}}

\def\II{{\mathds 1}}

\newcommand{\park}[1][k]{_{(#1)}}
\newcommand{\adaptedpar}[1]{\left(#1\right)}
\newcommand{\sign}{\operatorname{sign}}

\usepackage{natbib}